\documentclass[10pt,twocolumn,twoside]{IEEEtran}
\usepackage{multirow}
\usepackage{graphicx,subfigure}
\usepackage{amsmath,amssymb,amstext,amsfonts}
\usepackage{epsfig}
\usepackage{psfrag}
\usepackage{stfloats,color}
\usepackage{fancyhdr}
\usepackage[T1]{fontenc}
\usepackage{psfrag}
\usepackage{nccmath}
\usepackage{algorithm,algorithmic}
\usepackage{wasysym,tipa,dblaccnt,accents}
\usepackage{epstopdf}

\IEEEoverridecommandlockouts
\newcommand{\oprocendsymbol}{\hbox{$\bullet$}}
\newcommand{\oprocend}{\relax\ifmmode\else\unskip\hfill\fi\oprocendsymbol}

\renewcommand{\cdots}{\dots}
\newcommand{\until}[1]{\{1,\dots,#1\}}

\newcommand{\longthmtitle}[1]{\mbox{}\textit{(#1).}}
\renewcommand{\theenumi}{(\roman{enumi})}

\renewcommand{\baselinestretch}{1}

\DeclareMathOperator\E{\mathbb{E}}

\newtheorem{theorem}{Theorem}

\newtheorem{example}{Example}
\newtheorem{remark}{Remark}

\newtheorem{proposition}{Proposition}
\newtheorem{corollary}{Corollary}
\newtheorem{definition}{Definition}
\emergencystretch=\maxdimen
\input{tcilatex}

\begin{document}

\title{{\LARGE \textbf{Gramian-based reachability metrics for bilinear
      networks}}\thanks{A preliminary version of this work has been
    accepted as~\cite{YZ-JC:15-cdc} at the 2015 IEEE Conference on
    Decision and Control, Osaka, Japan.}}

\author{Yingbo Zhao and Jorge Cort\'{e}s \thanks{ Yingbo Zhao and
    Jorge Cort\'{e}s are with the Department of Mechanical and
    Aerospace Engineering, University of California at San Diego, La
    Jolla, CA 92093. Emails: \texttt{\small
      \{yiz326,cortes\}@ucsd.edu}}}

\maketitle

\begin{abstract}
  This paper studies Gramian-based reachability metrics for bilinear
  control systems. In the context of complex networks, bilinear
  systems capture scenarios where an actuator not only can affect the
  state of a node but also interconnections among nodes. Under the
  assumption that the input's infinity norm is bounded by some
  function of the network dynamic matrices, we derive a Gramian-based
  lower bound on the minimum input energy required to steer the state
  from the origin to any reachable target state.  This result
  motivates our study of various objects associated to the
  reachability Gramian to quantify the ease of controllability of the
  bilinear network: the minimum eigenvalue (worst-case minimum input
  energy to reach a state), the trace (average minimum input energy to
  reach a state), and its determinant (volume of the ellipsoid
  containing the reachable states using control inputs with no more
  than unit energy).  We establish an increasing returns property of
  the reachability Gramian as a function of the actuators, which in
  turn allows us to derive a general lower bound on the reachability
  metrics in terms of the aggregate contribution of the individual
  actuators. We conclude by examining the effect on the worst-case
  minimum input energy of the addition of bilinear inputs to
  difficult-to-control linear symmetric networks.  We show that the bilinear
  networks resulting from the addition of either inputs at a finite
  number of interconnections or at all self loops with weight
  vanishing with the network scale remain
  difficult-to-control. Various examples illustrate our results.
\end{abstract}

\thispagestyle{empty}\pagestyle{empty}

\section{Introduction}\label{Sec:Intro}

Complex networks such as electrical power grids, social networks, and
transportation networks, play an increasingly essential part in modern
society. A complex network typically consists of many dynamical
subsystems or nodes that interact with each other.  An important issue
is understanding to what extent the behavior of a large-scale, complex
network can be affected by controlling a few selected components.
Answering this question thoroughly would be of extreme value in the
analysis of biological networks and the design of engineered networks
with verifiable performance.  Existing results focus on linear control
models, where external control inputs can only directly affect the
state of a node, without affecting its interactions with other nodes.
In this paper, we are interested in taking the study of complex
networks to the nonlinear realm, where the control inputs may not only
affect directly node states but also change the interconnections among
nodes in the network.

\subsubsection*{Literature review}

Controllability refers to the property of being able to steer the
state of a dynamical system from any starting point to any terminal
point by means of appropriate inputs.  The controllability question in
the context of multi-agent systems and complex networks has recently
sparked an increasing body of research activity.  The basic idea is
understanding to what extent the state of the entire network can be
controlled by changing the states of some of its subsystems.  Using
graph-theoretic tools,~\cite{YYL-JJS-ALB:11} relates the number of
control nodes necessary to ensure controllability of a linear control
network to its degree distribution.~\cite{AO:14} considers the problem
of rendering a linear network controllable by affecting a small set of
variables with an external input.  The controllability properties of
consensus-type networks are studied employing the algebraic properties
of the network interconnection graph by~\cite{AR-MJ-MM-ME:09} in the
linear case and, more recently, by~\cite{CA-BG:14} in the nonlinear
case.  However, controllability is a binary, qualitative property that
does not quantify the amount of effort required to steer the system to
the terminal state. In the case of linear-time invariant systems, this
has motivated the study of various quantitative controllability
metrics based on the reachability Gramian\footnote{For a linear
  system, the reachability and the controllability Gramian are the
  same. However, this is not the case for bilinear systems.  Since we
  only discuss reachability, we use the term reachability Gramian.}.
\cite{GY-JR-YL-CL-BL:12} discusses upper and lower bounds on the
minimum energy to drive a network state from the origin to a target
state.  \cite{FP-SZ-FB:14} considers the selection of control nodes in
a complex linear network to reduce the worst-case minimum energy for
reachability.~\cite{THS-JL:14} proposes an optimal actuator placement
strategy in complex linear networks to reduce the average minimum
control energy over random target states.~\cite{VT-MAR-GJP-AJ:15}
considers the problem of minimal actuator placement in a linear
network so that a given bound on the minimum control effort for a
particular state transfer is satisfied while guaranteeing
controllability.

The use of linear control systems to model complex networks presumes
that the inputs only affect node states and not the interconnections
among them. This critical assumption may be too limiting for certain
classes of complex networks. For example, in the study of effective
connectivity in the brain, it is strongly
believed~\cite{KJF-LH-WP:03,JRI-AO-TM-MP-JS-GC-HP:14} that external
inputs not only have an effect on brain states in a particular area,
but can also change the strength of the coupling between the states of
different areas in the brain. These observations provide motivation
for our study of reachability metrics for complex networks modeled as
bilinear control systems.

Bilinear systems~\cite{CB-GD-GK:74,DE:09,PP-VY:10} are one of the
simplest classes of nonlinear systems but can be used to represent a
wide range of physical, chemical, economical, and biological systems
that cannot be effectively modeled using linear systems.
While reachability/controllability of bilinear systems as a binary
property has been widely investigated, see
e.g.,~\cite{DK-KN:85,UP-PF:92,TG-TT-JZ:73,
  LT-KC:11,DE:09} and references therein, few results are
available for quantitative metrics.  A notion of reachability Gramian
exists for bilinear systems, but its relation with the input energy
functional is not fully understood.
Under some assumptions, namely that at least one of the coefficient
matrices of the bilinear terms is nonsingular, that the target state
$x_{f}$ belongs to a neighborhood of the origin, and that an
integrability condition holds, \cite{WG-JM:98} shows that for a
continuous-time stable bilinear system with reachability Gramian
$\mathcal{W}_{c}$, the input energy required to drive the state from
the origin to $x_{f}$ is always greater than
$x_{f}^{T}\mathcal{W}_{c}^{-1}x_{f}$.  However, the integrability
condition may not hold for a general continuous-time bilinear system,
see \cite{EV:08,PB-TD:11} for a detailed discussion.  Instead of the
integrability condition,~\cite{PB-TD:11} assumes that the reachability
Gramian is diagonal and proves similar results for some $\epsilon>0$
and $x_{f}=\epsilon e_{j}$, where $e_{j}$ is any canonical unit vector
in $\mathbb{R}^{n}$.  However, for discrete-time bilinear systems,
there do not exist results analogous to these.

\subsubsection*{Statement of contributions}

We study the reachability properties of complex networks modeled as
bilinear control systems.  Our first contribution is the study of the
minimum input energy required to steer the system state from the
origin to any reachable target state.  Even though no closed-form
expression exists for the optimal controller and its associated cost
due to the nonlinear nature of bilinear systems, we establish a
Gramian-based lower bound on the minimum input energy required to
reach a target state, under the assumption that the infinity norm of
the input is bounded by some function of the system matrices.
Moreover, we show through a counterexample that this result does not
hold in general if the input is not constrained and, in fact, that
there does not exist a global positive lower bound for the ratio
between aggregate input norm and target state norm.  Our second
contribution introduces several Gramian-based reachability metrics for
bilinear control networks that quantify the worst-case and average
minimum input energy over all target states on the unit hypersphere in
the state space and the volume of the ellipsoid containing the
reachable states using control inputs with no more than unit energy.
We prove that the reachability Gramian, when viewed as a function of
the location of the actuators, exhibits an increasing returns
property. Building on this result, we derive a general lower bound on
the reachability metrics in terms of the aggregate contribution of the
individual actuators and lay out a greedy maximization strategy based
on selecting them sequentially starting with the one that has the
largest contribution.
Our third and final contribution involves bilinear systems built from
difficult to control linear networks. In particular, we show that a
bilinear system built from such a linear system by adding a finite
number of bilinear inputs is still difficult to control. We also
establish that a similar result holds even if the bilinear input can
equally affect all self loops in the network, with a strength that
vanishes with the network scale. Throughout the paper, we provide
numerous examples to illustrate the strengths and limitations of our
results.

\subsubsection*{Organization}

Section~\ref{Sec:PF} introduces discrete-time bilinear control systems
and states the problem of interest.  Section~\ref{Sec:pre} details
basic properties of the associated reachability Gramian and
Section~\ref{Sec:main} establishes its relationship with the input
energy functional. Motivated by this result,
Section~\ref{Sec:act_select} explores the problem of selecting
actuators to maximize various Gramian-based reachability metrics.
Section~\ref{Sec:imp_bl} examines the effect that the addition of
bilinear inputs has on the worst-case minimum input energy for
difficult-to-control linear networks.  We gather our conclusions and
ideas for future work in Section~\ref{Sec:con}.

\subsubsection*{Notation}

For a vector $x\in \mathbb{R}^{n}$, we use $x_{i}$ to denote its
$i$-th component and $\lVert x\rVert _{\infty }$ to denote its
infinity norm. For a matrix $M\in \mathbb{R}^{n\times m}$, we use
$M_{i}\in \mathbb{R}^{n}$ to denote its $i$-th column so that $M = [
M_{1} \hspace{1ex} M_{2} \hspace{1ex} \cdots \hspace{1ex} M_{m} ] $.
The vector generated by stacking the columns of $M$ is $\func{vec}(M)
= [ M_{1}^{T} \hspace{1ex} M_{2}^{T} \hspace{1ex} \cdots \hspace{1ex}
M_{m}^{T}]^{T}$.  The spectral norm (maximum singular value) of $M$ is
denoted by $\lVert M\rVert $.  For symmetric (square) matrices, we use
$\lambda _{\max }(M)$ to denote the maximum eigenvalue and $M>0$
(resp. $M\geq 0$) to denote that $M$ is positive definite (resp. $M$
is positive semidefinite).  The spectral radius of $M$, denoted $\rho
(M)$, is the supremum among the magnitudes of its eigenvalues. The
matrix $M$ is Schur stable if $\rho (M)<1$.  We let $ \mathbf{0}_{n}$
and $\mathbf{0}_{m\times n}$ denote the $n$-vector and $m\times n$
matrix with all elements equal to zero, respectively. We let $I_{n}$
denote the identity matrix of dimension $n\times n$.  Given a sequence
$\{x(k)\}_{k =0}^\infty$ and $j_1\leq j_2\in \mathbb{Z}_{\geq 0}$, we
use $\{x\}_{j_1}^{j_2}$ to denote the finite sequence $\{
x(j_1),x(j_1+1),\dotsc ,x(j_2) \}$. We omit $j_1$ if $j_1=0$. We let
$\func{diag}(A_{1},\cdots ,A_{n})$ denote the block-diagonal matrix
defined by the matrices $A_1,\dots,A_n$. Finally, the symbol $\otimes
$ represents the Kronecker product of matrices.

\section{Problem Formulation}\label{Sec:PF}

We consider the class of discrete-time bilinear control systems with
state-space representation%
\begin{equation}
  x(k+1) = A x(k)+\sum_{j=1}^{m} (F_{j}x(k)+B_{j}) u_{j}(k),  \label{eq:BLS}
\end{equation}%
where $k\in \mathbb{Z}_{\geq 0}$ is the time index, $x(k)\in
\mathbb{R}^{n}$ is the system state, $u(k) = [u_{1}(k),\dots,u_{m}(k)]
\in \mathbb{R}^m$ is the control input and $ A$, $F_{j}\in
\mathbb{R}^{n\times n}$, $B_{j}\in \mathbb{R}^{n}$, $j \in \until{m}$
are the system matrices. When convenient, we simply refer to the
bilinear control system~\eqref{eq:BLS} by~$(A,F,B)$, where $ F = [
F_{1} \hspace*{1ex} F_{2} \hspace*{1ex} \cdots \hspace*{1ex} F_{m} ]$
and $B = [ B_{1} \hspace*{1ex} B_{2} \hspace*{1ex} \cdots
\hspace*{1ex} B_{m}]$.  Throughout the paper, we assume that $A$ is
Schur stable.  There is no loss of generality in letting the same
input $u_{j}(k)$ appear simultaneously in the bilinear and linear
terms in~\eqref{eq:BLS}.  In fact, a general bilinear system%
\begin{equation}
  x(k+1) = A x(k) + \sum_{j=1}^{p}\bar{F}_{j}x(k)v_{j}(k) +
  \sum_{j=1}^{q}\bar{B}_{j}w_{j}(k) ,
  \label{eq:BLS_1}
\end{equation}
with $v(k)\in \mathbb{R}^{p}$ and $w(k)\in \mathbb{R}^{q}$, can be
rewritten in the form of~\eqref{eq:BLS} by defining $u(k) = [ v^{T}(k)
\hspace*{1ex} w^{T}(k) ]^{T}$, $ F =[ \bar{F} \hspace*{1ex}
\mathbf{0}_{n\times nq} ]$, $B = [ \mathbf{0}_{n\times p}
\hspace*{1ex} \bar{B} ]$, and $m=p+q$.

The system~\eqref{eq:BLS} is \emph{controllable} in a set $\mathbb{%
  S\subseteq R}^{n}$ if, for any given pair of initial and target
states in $ \mathbb{S}$, there exists a finite control sequence that
drives the system from one to the other. The notion of
\emph{reachability} corresponds to controllability from the origin,
i.e., the existence of a finite control sequence that takes the state
from the origin to an arbitrary target state in
$\mathbb{S}$. Controllability and reachability are qualitative
measures of a system that do not precisely characterize how easy or
difficult, in terms of control effort, it is for the system to go from
one state to another.  Our objective is to provide quantitative
measures of the degree of reachability for the bilinear control
system~\eqref{eq:BLS}. Note that, unlike linear systems, the
controllability of a bilinear system depends on its initial
condition. Here, we focus on reachability. Formally, consider the
minimum-energy optimal control problem for a given target state
$x_{f}$ and a time horizon $K\in \mathbb{Z}_{>0}$, defined by
\begin{equation}\label{pb:original_op}
  \begin{array}{rl}
    \min_{\{u\}^{K-1}} & \sum_{k=0}^{K-1}u^{T}(k)u(k) \\
    \text{s.t.}
    & \multicolumn{1}{l}{\text{$\eqref{eq:BLS}$ holds }
      \forall k=0,\ldots ,K-1,}
    \\
    & \multicolumn{1}{l}{x(0)=\mathbf{0}_{n}, \,\, x(K)=x_{f}.}%
  \end{array}%
\end{equation}
%
Our aim can then be formulated as seeking to characterize the value of
the optimal solution of~\eqref{pb:original_op} in terms of the data
$(A,F,B)$ that defines the bilinear control system.

\section{Reachability Gramian}\label{Sec:pre}

This section introduces the notion of reachability Gramian for stable
discrete-time bilinear systems and characterizes some useful
properties. Our discussion sets the basis for our later analysis on
the relationship between the reachability Gramian and the
minimum-energy optimal control problem~\eqref{pb:original_op}.


\begin{definition}\longthmtitle{Reachability
    Gramian~\cite{LZ-JL-BH-GY:03}}
  The reachability Gramian for a stable discrete-time bilinear system
  $(A,F,B)$ is
  \begin{equation}
    \mathcal{W} = \sum_{i=1}^{\infty }\mathcal{W}_{i},  \label{eq:reach_Gramian}
  \end{equation}%
  where
  \begin{align*}
    \mathcal{W}_{i}& =\sum_{k_{1},\ldots ,k_{i}=0}^{\infty
    }\mathcal{P}%
    _{i}(\{k\}_{1}^{i})\mathcal{P}_{i}^{T}(\{k\}_{1}^{i}), \\
    \mathcal{P}_{1}(\{k\}_{1}^{1})& =A^{k}B\in \mathbb{R}^{n\times m}, \\
    \mathcal{P}_{i}(\{k\}_{1}^{i})& =A^{k_{i}}F(I_{m}\otimes
    \mathcal{P}%
    _{i-1}(\{k\}_{1}^{i-1}))\in \mathbb{R}^{n\times m^{i}},\,\, i\geq
    2.
  \end{align*}
\end{definition}

The reachability Gramian for continuous-time bilinear systems is
defined analogously, see e.g.,~\cite{SB-MB-US:94,WG-EV:06}. This
notion of reachability Gramian is widely used in model order reduction
of bilinear systems~\cite{LZ-JL:02,PB-TB-TD:11} and linear switched
systems \cite{MP-RW-JL:13}. Notice that, for linear control systems
(i.e., $F=\mathbf{0}_{n\times nm}$ in~\eqref{eq:BLS}), the
reachability Gramian in~\eqref{eq:reach_Gramian} takes the form
\begin{equation}
  \mathcal{W}=\mathcal{W}_{1}=\sum_{k=0}^{\infty }A^{k}BB^{T}(A^{T})^{k},
\label{eq:reach_Gramian_linear}
\end{equation}%
which is the reachability Gramian associated to the corresponding
discrete-time linear time-invariant system~\cite{TK:80}.

Throughout the paper, we assume that $(A,F,B)$ are such that the
series in~\eqref{eq:reach_Gramian} converges and the resulting matrix
is positive definite. A sufficient condition for the latter is that
$(A,\mathbf{0}_{n\times nm},B)$ is controllable, which in turn is
equivalent to $\mathcal{W}_{1}>0$.  We discuss necessary and
sufficient conditions for the convergence of the series below
in~\eqref{ineq:exist_W}.


The reachability Gramian is a solution of a generalized Lyapunov
equation~\cite{WG-JM:98,WG-EV:06}. The next result appears
in~\cite{LZ-JL-BH-GY:03,PB-TD:11}.  We provide a formal proof for the
sake of completeness.

\begin{theorem}\longthmtitle{Generalized Lyapunov
    equation}\label{Theo1}
  The reachability Gramian $\mathcal{W}$ satisfies the following
  generalized Lyapunov equation
  \begin{equation}
    A\mathcal{W}A^{T} - \mathcal{W} + \sum_{j=1}^{m}F_{j}\mathcal{W}%
    F_{j}^{T}+BB^{T}= \mathbf{0}_{n\times n}.  \label{eq:Theo_1}
  \end{equation}
\end{theorem}
\begin{proof}
  From~\eqref{eq:reach_Gramian_linear}, one can see that
  $\mathcal{W}_{1}$ satisfies
  \begin{equation}\label{eq1:Pf_Theo_1}
    A\mathcal{W}_{1}A^{T}-\mathcal{W}_{1}+BB^{T}=\mathbf{0}_{n\times n}.
  \end{equation}%
  For $i\geq 2$, we obtain
  \begin{align}
    \mathcal{W}_{i}& =\sum_{k_{1},\ldots ,k_{i}=0}^{\infty
    }\mathcal{P}%
    _{i}(\{k\}_{1}^{i})\mathcal{P}_{i}^{T}(\{k\}_{1}^{i})  \notag \\
    & =\sum_{k_{1},\ldots ,k_{i}=0}^{\infty }A^{k_{i}}F\left(I_{m}\otimes
    \mathcal{P}_{i-1}\mathcal{P}_{i-1}^{T})F^{T}(A^{k_{i}}\right)^{T}  \notag \\
    & =\sum_{k_{i}=0}^{\infty
    }A^{k_{i}}\bigl(\sum_{j=1}^{m}F_{j}\sum_{k_{1},%
      \ldots ,k_{i-1}=0}^{\infty
    }\mathcal{P}_{i-1}\mathcal{P}_{i-1}^{T}F_{j}^{T}%
    \bigr)(A^{k_{i}})^{T}  \notag \\
    & =\sum_{k_{i}=0}^{\infty
    }A^{k_{i}}\bigl(\sum_{j=1}^{m}F_{j}\mathcal{W}%
    _{i-1}F_{j}^{T}\bigr)(A^{k_{i}})^{T}.  \label{eq:W_i_W_i-1}
  \end{align}%
  Therefore,
  \begin{equation}
    A\mathcal{W}_{i}A^{T}-\mathcal{W}_{i}+\sum_{j=1}^{m}F_{j}\mathcal{W}%
    _{i-1}F_{j}^{T}=\mathbf{0}_{n\times n}.  \label{eq2:Pf_Theo_1}
  \end{equation}%
  We obtain (\ref{eq:Theo_1}) by summing\ (\ref{eq1:Pf_Theo_1}) and
  (\ref{eq2:Pf_Theo_1}) with $i$ ranging from $2$ to $\infty $.
\end{proof}

It is thus possible to obtain the reachability Gramian $\mathcal{W}$
by solving the generalized Lyapunov equation (\ref{eq:Theo_1}), which
one can do by computing%
\begin{equation}
  \func{vec}(\mathcal{W})=(I_{n^{2}}-A\otimes A-\sum_{j=1}^{m}F_{j}\otimes
  F_{j})^{-1}\func{vec}(BB^{T}).  \label{eq:sol_Lya}
\end{equation}%
Moreover, \cite{RGA-EIJ:71} shows that a unique positive semi-definite
solution $\mathcal{W}$\ exists if and only if%
\begin{equation}
  \rho (A\otimes A+\sum_{j=1}^{m}F_{j}\otimes F_{j})<1,  \label{ineq:exist_W}
\end{equation}%
a condition that we assume to hold throughout the paper.

\begin{remark}\longthmtitle{Connection with mean-square stability of stochastic
    bilinear systems}
  Following~\cite{AH-RES:87,RES-SMK-EY:91}, consider the
  time-invariant discrete-time stochastic bilinear system
  \begin{equation}
    x(k+1) = A x(k) + \sum_{j=1}^{p}F_{j}x(k)v_{j}(k)+\sum_{j=1}^{q} B_{j}w_{j}(k),
    \label{eq:sto_bilinear}
  \end{equation}%
  where $v(k)\in \mathbb{R}^{p}$ and $w(k)\in \mathbb{R}^{q}$ are
  random variables. We have used the form \eqref{eq:BLS_1}, which is
  equivalent to~\eqref{eq:BLS}. Assume $\{w\}^{\infty }$ and
  $\{v\}^{\infty }$ are uncorrelated stationary zero-mean white
  processes satisfying
  \begin{equation*}
    \E\lbrack v(k)v^{T}(j)]=I_{p}\delta _{kj},\quad \E\lbrack
    w(k)w^{T}(j)] = I_{q}\delta _{kj}.
  \end{equation*}%
  If the system is mean-square stable, then the positive semi-definite
  steady state covariance $\E\lbrack x(k)x^{T}(k)]$ satisfies the
  generalized Lyapunov equation~\eqref{eq:Theo_1}. Therefore, the
  existence of the reachability Gramian is related to the mean square
  stability of the corresponding stochastic bilinear system
  \eqref{eq:sto_bilinear}, which is equivalent to
  \eqref{ineq:exist_W}.\relax%
  \ifmmode\else\unskip\hfill\fi\hbox{$\bullet$}
\end{remark}

To conclude this section, we show that any target state $x_{f}$ that is
reachable from the origin, $x(0)=\mathbf{0}_{n}$, belongs to $\func{Im}(%
\mathcal{W})$. An analogous result is known for continuous-time bilinear
systems~\cite[Theorem 3.1]{PB-TD:11}.

\begin{proposition}\label{Theo2_ImW_invariant}
  The subspace $\func{Im}(\mathcal{W})$ is invariant under the
  bilinear control system~\eqref{eq:BLS} defined by $%
  (A,F,B)$.
\end{proposition}

\begin{proof}
For all $v\in \func{Ker}(\mathcal{W})$, it holds that
\begin{equation*}
0=v^{T}\mathcal{W}v=v^{T}(A\mathcal{W}A^{T}+\sum_{j=1}^{m}F_{j}\mathcal{W}%
F_{j}^{T}+BB^{T})v,
\end{equation*}%
where the last equation follows from (\ref{eq:Theo_1}). As a result,
\begin{equation*}
A^{T}v\in \func{Ker}(\mathcal{W}),\quad F_{j}^{T}v\in \func{Ker}(\mathcal{W}%
),\quad B^{T}v=0.
\end{equation*}%
Note that since $\mathcal{W}$ is symmetric, $\func{Im}(\mathcal{W})=(\func{%
Ker}(\mathcal{W}))^{\bot }$. Therefore, if $x(k)\in \func{Im}(\mathcal{W})$,
then
\begin{equation*}
x^{T}(k+1)v=x^{T}(k)A^{T}v+%
\sum_{j=1}^{m}u_{j}(k)(x^{T}(k)F_{j}^{T}v+B_{j}^{T}v)=0,
\end{equation*}%
which implies that $x(k+1)\in \func{Im}(\mathcal{W})$ because $x(k+1)$ is
orthogonal to all $v\in \func{Ker}(\mathcal{W})$ and the proof is complete.
\end{proof}

Given that $x(0) = \mathbf{0}_{n}\in \func{Im}(\mathcal{W})$,
Theorem~\ref{Theo2_ImW_invariant} implies that $x(k)\in \func{Im}(%
\mathcal{W})$ for all $k\in \mathbb{Z}_{\geq 0}$, and therefore, any
target state $x_{f}$ that is reachable from the origin belongs
to~$\func{Im}(%
\mathcal{W})$.

\section{Minimum input energy for reachability}\label{Sec:main}

In this section, we obtain a lower bound on the minimum input energy
required to steer the state of a bilinear control system from the
origin to any reachable state under the assumption that the input norm
is upper bounded. The bound on the minimum input energy is a function
of the reachability Gramian. We build on this result later to define
reachability metrics for bilinear control systems.

From the formulation~\eqref{pb:original_op} of the optimal control
problem in Section~\ref{Sec:PF}, the necessary optimality conditions
for the solution $\{u^{\ast }\}^{K-1}$ lead to the following nonlinear
two-point boundary-value problem~\cite{ZA-ZG:94} for $k=0,\ldots ,K-1$
\begin{align}
  x(k+1)& =Ax(k)+\frac{1}{2}\sum_{j=1}^{m}(F_{j}x(k)+B_{j})
  (F_{j}x(k)+B_{j})^{^{T}} \hspace*{-5pt} \eta (k), \notag
  \\
  \eta (k-1)& =A^{T}\eta (k)+\frac{1}{2}\sum_{j=1}^{m}\eta(k)^{^T}
  \bigl(%
  F_{j}x(k) +B_{j}\bigr) \cdot F_{j}^{^{T}}\eta (k) , \notag
  \\
  u_{j}^{\ast }(k)& =\frac{1}{2}(F_{j}x(k)+B_{j})^{^{T}}\eta (k).
  \label{eq:nonlin_2pt_bv}
\end{align}
For a stable, controllable, linear time-invariant system
$(A,\mathbf{0}%
_{n\times nm},B)$, one can obtain analytically the optimal control
sequence from~\eqref{eq:nonlin_2pt_bv},
\begin{equation*}
u^{\ast }(k)=B^{^{T}}(A^{^{T}})^{K-k-1}\mathcal{W}_{1,K}^{-1}x_{f},
\end{equation*}%
with associated minimum control energy
\begin{equation}
\sum_{k=0}^{K-1}(u^{\ast }(k))^{T}u^{\ast }(k)=x_{f}^{T}\mathcal{W}%
_{1,K}^{-1}x_{f}>x_{f}^{T}\mathcal{W}_{1}^{-1}x_{f},  \label{eq:LTI_Gra_LB}
\end{equation}%
where $\mathcal{W}_{1,K}\triangleq
\sum_{k=0}^{K-1}A^{k}BB^{T}(A^{T})^{k}$ denotes the $K$-step
controllability Gramian of the linear time-invariant system. In
general, the nonlinear two-point boundary-value problem~
\eqref{eq:nonlin_2pt_bv} does not admit an analytical solution, which
has motivated the use of numerical approaches such as successive
approximations~\cite{GYT-HM-BLZ:05} and iterative
methods~\cite{EH-BT:88}. Given the paper goals, we do not try to find
the optimal control sequence but instead focus on the expression for
the minimum control energy and, specifically, on its connection with
the reachability Gramian.

The next result shows how, when the infinity norm of the input is
upper bounded by a specific function of the system matrices, the lower
bound in~\eqref{eq:LTI_Gra_LB} also holds.

\begin{theorem}\longthmtitle{The reachability Gramian is a metric for
    reachability}\label{Theo5}
  For the bilinear control system~\eqref{eq:BLS}, define
  \begin{align*}
    \beta & \triangleq -\sum_{j=1}^{m}\lVert A^{T}\Psi
    F_{j}+F_{j}^{T}\Psi A\lVert +\Bigl(\bigl(\sum_{j=1}^{m}\lVert
    A^{T}\Psi F_{j}+F_{j}^{T}\Psi A\lVert \bigr)^{2}
    \\
    & \quad -4\sum_{i,j=1}^{m}\lVert F_{j}^{T}\Psi F_{i}\lVert \cdot
    \lambda _{\max }(A^{T}\Psi A-\mathcal{W}^{-1})\Bigr)^{1/2},
    \\
    \Psi & \triangleq
    \mathcal{W}^{-1}-\mathcal{W}^{-1}B(B^{T}\mathcal{W}%
    ^{-1}B-I_{m})^{-1}B^{T}\mathcal{W}^{-1} .
  \end{align*}%
  For $K\in \mathbb{Z}_{\geq 1}$, if
  \begin{equation}
    \lVert u(k)\rVert _{\infty }\leq 2^{-1}\bigl(\sum_{i,j=1}^{m}\lVert
    F_{j}^{T}\Psi F_{i}\rVert \bigr)^{-1}\beta ,  \label{eq:Theo5_1}
  \end{equation}%
  for all $k=0,1,\cdots ,K-1$,
  then
  \begin{equation}
    \sum_{k=0}^{K-1}u^{T}(k)u(k)\geq x^{T}(K)\mathcal{W}^{-1}x(K).
    \label{eq:Theo5_2}
  \end{equation}
\end{theorem}
\begin{proof}
  We consider the Lyapunov functional $V(x)=x^{T}\mathcal{W}^{-1}x$ and obtain
  \begin{align}
    & V(x(k+1))-V(x(k))-u^{T}(k)u(k)  \notag \\
    & \quad =\left[
      \begin{array}{c}
        x(k) \\
        u(k)%
      \end{array}%
    \right] ^{T}\left[
      \begin{array}{cc}
        \Phi _{11}(k) & \Phi _{21}^{T}(k) \\
        \Phi _{21}(k) & \Phi _{22}(k)%
      \end{array}%
    \right] \left[
      \begin{array}{c}
        x(k) \\
        u(k)%
      \end{array}%
    \right] ,  \label{eq1:Pf_Theo5}
  \end{align}%
  where
  \begin{align*}
    \Phi _{11}(k)&
    =A^{T}\mathcal{W}^{-1}A-\mathcal{W}^{-1}\mathcal{+}%
    \sum_{i,j=1}^{m}u_{j}(k)F_{j}^{T}\mathcal{W}^{-1}F_{i}u_{i}(k) \\
    & \quad
    +\sum_{j=1}^{m}(A^{T}\mathcal{W}^{-1}F_{j}+F_{j}^{T}\mathcal{W}%
    ^{-1}A)u_{j}(k)\in \mathbb{R}^{n\times n},
    \\
    \Phi _{21}(k)& =B^{T}\mathcal{W}^{-1}A+B^{T}\mathcal{W}^{-1}%
    \sum_{j=1}^{m}F_{j}u_{j}(k)\in \mathbb{R}^{n\times m}, \\
    \Phi _{22}(k)& =B^{T}\mathcal{W}^{-1}B-I_{m}\in
    \mathbb{R}^{m\times m}.
  \end{align*}%
  In the rest of the proof, we show that the matrix $\Phi (k)=[\Phi
  _{ij}(k)]\in \mathbb{R}^{(n+m)\times (n+m)}\leq 0$ under
  (\ref{eq:Theo5_1}).  First, multiplying the generalized Lyapunov
  equation (\ref{eq:Theo_1}) by the vector $\mathcal{W}^{-1}B$ from
  the right-hand side, and by the vector $B^{T}%
  \mathcal{W}^{-1}$ from the left-hand side, we obtain after some
  manipulation
  \begin{align}
    \Phi _{22}(k)&
    =-(B^{T}\mathcal{W}^{-1}B)^{-1}B^{T}\mathcal{W}^{-1}\bigl(A%
    \mathcal{W}A^{T}  \notag \\
    & \quad
    +\sum_{j=1}^{m}F_{j}\mathcal{W}F_{j}^{T}\bigr)\mathcal{W}^{-1}B<0,
    \label{eq2:Pf_Theo5}
  \end{align}%
  where we have used the fact that $\mathcal{W}$ is positive
  definite. Moreover, it follows that
  \begin{align}
    & \Phi _{11}(k)-\Phi _{21}^{T}(k)\Phi _{22}^{-1}(k)\Phi _{21}(k)
    \notag \\
    & \quad =\sum_{i,j=1}^{m}u_{j}(k)F_{j}^{T}\Psi
    F_{i}u_{i}(k)+A^{T}\Psi A-%
    \mathcal{W}^{-1}  \notag \\
    & \quad \quad +\sum_{j=1}^{m}(A^{T}\Psi F_{j}+F_{j}^{T}\Psi
    A)u_{j}(k) \notag
    \\
    & \quad \leq \bigl(\sum_{i,j=1}^{m}\lVert F_{j}^{T}\Psi
    F_{i}\lVert \cdot \lVert u(k)\lVert _{\infty }^{2}+\lambda _{\max
    }(A^{T}\Psi A-\mathcal{W}%
    ^{-1}) \notag
    \\
    & \quad \quad +\sum_{j=1}^{m}\lVert A^{T}\Psi F_{j}+F_{j}^{T}\Psi
    A\lVert \cdot \lVert u(k)\lVert _{\infty }\bigr)I_{n} \leq
    0, \label{eq3:Pf_Theo5}
  \end{align}%
  where the last inequality holds because of (\ref{eq:Theo5_1}). Using
  the Schur complement lemma~\cite{SB-LV:09}, (\ref{eq2:Pf_Theo5}) and
  (\ref{eq3:Pf_Theo5}) imply $\Phi (k)\leq 0$. Finally,
  summing~(\ref{eq1:Pf_Theo5}) with respect to $k=0,1,\cdots ,K-1$ and
  noting $V(x(0))=0$, we get~(\ref{eq:Theo5_2}).
\end{proof}

The sufficient condition (\ref{eq:Theo5_1}) is a magnitude constraint
at every actuator. Theorem \ref{Theo5} provides a reachability
Gramian-based lower bound on the minimum input energy required to
drive the state from the origin to any reachable state.
There are two reasons why this bound
may be conservative. First, instead of considering the sign of the sum
over the entire time horizon $k=0,1,\cdots ,K-1$, the proof's strategy
relies on each individual inequality
\begin{equation*}
  x^{T}(k+1) \mathcal{W}^{-1}x(k+1)-x^{T}(k)\mathcal{W}^{-1}x(k)\leq
  u^{T}(k)u(k)
\end{equation*}%
to hold for every time step $k$. Second, the bounding in
inequality~\eqref{eq3:Pf_Theo5} may introduce conservativeness.

\begin{remark}\longthmtitle{Positivity of the input upper bound
    in~\eqref{eq:Theo5_1}}\label{Remark1}
  From the definition of $\beta $ in Theorem~\ref{Theo5}, it is clear
  that the upper bound in (\ref{eq:Theo5_1}) on the infinity norm of
  the input is positive if and only if the matrix
  $\mathcal{G}(A,F,B)=A^{T}\Psi A-\mathcal{W }^{-1}$ is negative
  definite. We have computed the upper bound for hundreds of randomly
  generated matrix tuples $(A,F,B)$ and they all turn out to be
  positive. However, we have not been able to establish analytically
  the negative definiteness of $ \mathcal{G}$ in general due to its
  complex dependence on $A,F,B$. This fact can be established directly
  for the class of scalar bilinear systems. \oprocend
\end{remark}

\begin{corollary}
  \mbox{}\textit{(Scalar case for Theorem \ref{Theo5}).}\label{Theo4}
  Consider the class of scalar bilinear systems $(a,f,b)$:
  \begin{equation}
    x(k+1)=ax(k)+fx(k)u(k)+bu(k).  \label{BLS:scalar}
  \end{equation}%
  If $\forall k=1,2,\cdots ,K$,
  \begin{equation}
    |u(k)+af^{-1}|\leq \sqrt{a^{2}f^{-2}+1},  \label{ub_control_sca_BL}
  \end{equation}%
  then%
  \begin{equation*}
    \sum_{k=0}^{K-1}u^{2}(k)\geq \mathcal{W}^{-1}x^{2}(K).
  \end{equation*}
\end{corollary}
\begin{proof}
  For a scalar bilinear system (\ref{BLS:scalar}), we immediately have
  $%
  \mathcal{W}=(1-a^{2}-f^{2})^{-1}b^{2}$, either from the reachability
  Gramian definition (\ref{eq:reach_Gramian}) or from the generalized
  Lyapunov equation (\ref{eq:Theo_1}). Using the Lyapunov function
  $V(x)=x^{T} \mathcal{W}^{-1}x$, we obtain after some manipulation,
  \begin{align}
    & V(x(k+1))-V(x(k))-u^{2}(k) \notag
    \\
    \quad & =-(1-(a+fu(k))^{2})\mathcal{W}^{-1}x^{2}(k) \notag
    \\
    & \quad \quad
    +2b(a+fu(k))\mathcal{W}^{-1}u(k)x(k)-(1-b^{2}\mathcal{W}%
    ^{-1})u^{2}(k) \notag
    \\
    \quad & \overset{(a)}{\leq
    }-(|b|\mathcal{W}^{-1}|x(k)|-\sqrt{a^{2}+f^{2}}%
    |u(k)|)^{2}\leq 0, \label{eq1:Pf_Theo4}
  \end{align}%
  where $(a)$ holds because of (\ref{ub_control_sca_BL}). By summing
  inequality (\ref{eq1:Pf_Theo4}) with respect to $k=1,2,\cdots ,K-1$
  and noting that $V(x(0))=0$, the proof is complete.
\end{proof}

We end this section with two examples to complement the result in
Theorem~\ref{Theo5}. First, we show through a counter example that the
inequality~\eqref{eq:Theo5_2} does not hold in general if the input
norm is unconstrained. In fact, there does not exist a global lower
bound for
\begin{align*}
  \frac{\sum_{k=0}^{K-1}u^{T}(k)u(k)}{\lVert x(K)\lVert ^{2}}
\end{align*}
that is strictly greater than $0$.

\begin{example}\longthmtitle{There is no positive global lower bound for $
    \sum_{k=0}^{K-1}u^{T}(k)u(k)/\lVert x(K)\lVert^{2}$}\label{Example2}
  Consider the $2$-step reachability problem for the scalar bilinear
  system $(a,f,1)$,
  \begin{align}
    x(k+1)& =ax(k)+fx(k)u(k)+u(k), \notag
    \\
    x(0)& =0,\quad x(2)=x_{f}.  \label{BLS:scalar_example}
  \end{align}%
  It is easy to obtain from (\ref{BLS:scalar_example}) that
  \begin{equation*}
    u(0)=(a+fu(1))^{-1}(x_{f}-u(1)).
  \end{equation*}%
  By denoting $x_{f}=Mu(1)$ with $M\in $ $\mathbb{R}$, we have for any positive scalar $w$,
  \begin{align*}
    & u^{2}(0)+u^{2}(1)-w^{-1}x_{f}^{2}
    \\
    & \quad
    =\bigl(((a+fu(1))^{-2}-w^{-1})M^{2}-2(a+fu(1))^{-2}M
    \\
    & \quad \quad +1+(a+fu(1))^{-2}\bigr)u^{2}(1).
  \end{align*}%
  Choosing $u(1)$ large enough such that $(a+fu(1))^{2}>w$, there exists $M$ such that
  \begin{equation*}
    u^{2}(0)+u^{2}(1)-w^{-1}x_{f}^{2}<0.
  \end{equation*}
  Therefore, there exists $x_{f}$, $u(0)$, $u(1)$ such that under the
  dynamics~\eqref{BLS:scalar_example}, $u^{2}(0) +
  u^{2}(1)<w^{-1}x_{f}^{2}$ for any $w>0$.
  \oprocend
\end{example}

Our second example illustrates the tightness of the Gramian-based
lower bound~\eqref{eq:Theo5_2} for the input energy functional.

\begin{example}\longthmtitle{Tightness of the Gramian-based lower bound in
    Theorem~\ref{Theo5}}\label{Example3}
  Consider the following single-input bilinear control system taken
  from~\cite{TH-SM:84},
  \begin{equation}
    (A,f,b):x(k+1)=Ax(k)+fu(k)x(k)+bu(k),  \label{eq1:Ex3}
  \end{equation}%
  where
  \begin{align}
    A& =\left[
      \begin{array}{ccccc}
        0 & 0 & 0.024 & 0 & 0 \\
        1 & 0 & -0.26 & 0 & 0 \\
        0 & 1 & 0.9 & 0 & 0 \\
        0 & 0 & 0.2 & 0 & -0.06 \\
        0 & 0 & 0.15 & 1 & 0.5%
      \end{array}%
    \right] ,\quad b=\left[
      \begin{array}{c}
        0.8 \\
        0.6 \\
        0.4 \\
        0.2 \\
        0.5%
      \end{array}%
    \right] ,  \notag \\
    f& =\func{diag}(0.1,0.2,0.3,0.4,0.5).  \notag
  \end{align}%
  We use (\ref{eq:sol_Lya}) to compute the reachability Gramian
  $\mathcal{W}$ as
  \begin{equation*}
    \left[
      \begin{array}{ccccc}
        0.6505 & 0.4572 & 0.4741 & 0.1945 & 0.5342 \\
        0.4572 & 1.2846 & -0.4169 & -0.1165 & -0.3682 \\
        0.4741 & -0.4169 & 6.9412 & 1.1619 & 4.5490 \\
        0.1945 & -0.1165 & 1.1619 & 0.2708 & 0.9262 \\
        0.5342 & -0.3682 & 4.5490 & 0.9262 & 5.2681%
      \end{array}%
    \right] .
  \end{equation*}%
  Inequality (\ref{eq:Theo5_1}) provides an upper bound on $\lVert
  u(k)\rVert _{\infty }$,
  \begin{equation}
    \lVert u(k)\rVert _{\infty }\leq 0.0011.  \label{eq:Example1}
  \end{equation}%
  \begin{figure}[tbh]
    \centering
    \includegraphics[scale=0.5]{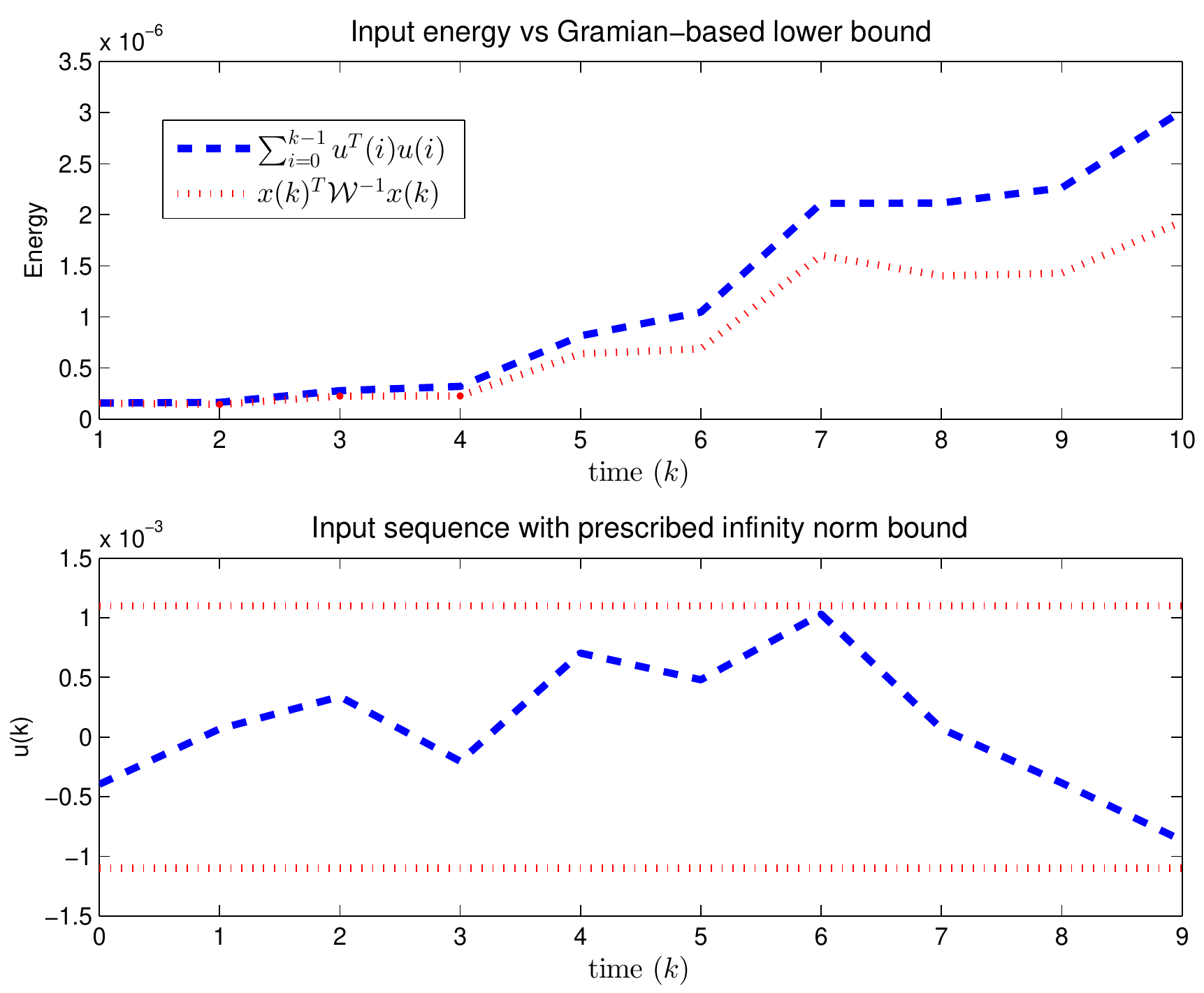}
    \caption{Input energy vs Gramian-based lower bound.}
    \label{Fig1}
  \end{figure}
  Figure~\ref{Fig1} compares the input energy functional $%
  \sum_{i=0}^{k-1}u^{2}(i)$ with the Gramian-based lower bound
  $x^{T}(k)%
  \mathcal{W}^{-1}x(k)$ for $k\leq K=10$ and an arbitrarily chosen
  input sequence $\{u\}^{K-1}$ satisfying~\eqref{eq:Example1}. Since
  the gap between the minimum input energy and the lower bound cannot
  be greater than the one shown in the plot, Figure~\ref{Fig1} shows
  that the Gramian-based lower bound is a good estimate of the minimum
  input energy required to drive the state from the origin to another
  state. \oprocend
\end{example}

\section{Reachability metrics for bilinear
  networks}\label{Sec:act_select}

The inequality~\eqref{eq:Theo5_2} connecting the reachability Gramian
and the minimum energy required to steer the system from the origin to
an arbitrary terminal state allows us to extend the reachability
metrics defined for complex linear systems
in~\cite{FP-SZ-FB:14,GY-JR-YL-CL-BL:12,THS-JL:14} to bilinear control
systems. We therefore consider the minimum eigenvalue, the trace, and
the determinant of the Gramian as reachability
metrics. The minimum eigenvalue $\lambda _{\func{min}}(\mathcal{W})$
characterizes the \emph{minimum input energy required in the worst
  case} to reach a state.  Given the observation,
cf.~\cite{THS-JL:14}, that
\begin{align*}
  \frac{\int_{\{x\in \mathbb{R}^{n}|\lVert x\lVert
      =1\}}x^{T}\mathcal{W}^{-1}xdx}{\int_{\{x\in
      \mathbb{R}^{n}|\lVert x\lVert =1\}}dx} = \frac{\func{tr}
    (\mathcal{W}^{-1})}{n}\geq \frac{n}{\func{tr}(\mathcal{W})},
\end{align*}
the trace $\func{tr}(\mathcal{W})$ characterizes the \emph{average
  minimum control energy} required to reach a state. Finally, the
determinant $\det (\mathcal{W})$ reflects the volume of the ellipsoid
containing the reachable states using inputs with no more than unit
energy.

Formally, our goal is to choose $m$ actuators from a given group of
$M$ candidates ($m\leq M$) such that $\lambda
_{\func{min}}(\mathcal{W})$, $%
\func{tr}(\mathcal{W})$ or $\det (\mathcal{W})$ is maximized,
depending on the specific objective at hand. Denoting $V=\{1,\cdots
,M\}$ and $ S=\{s_{1},\cdots ,s_{m}\}$, we write this combinatorial
optimization problem as%
\begin{equation}
  \max_{S\subseteq V} \; f(\mathcal{W}(S)),  \label{PB:max_tr}
\end{equation}%
where $f:\mathbb{R}^{n\times n}\rightarrow \mathbb{R}_{\geq 0}$ can be
$ \func{tr}(\mathcal{\cdot })$, $\lambda _{\func{min}}(\mathcal{\cdot
})$ or $ \det (\mathcal{\cdot })$. We use $\mathcal{W}(S)$ instead of
$\mathcal{W}$ to indicate its dependence on the choice
of~$S$. Similarly, we denote the input matrices $B$ and $F$ as
$B_{S}=[ b_{s_{1}} \hspace*{1ex} \cdots \hspace*{1ex} b_{s_{m}} ] $
and $F_{S}=[ F_{s_{1}} \hspace*{1ex} \cdots \hspace*{1ex} F_{s_{m}}
]$, respectively, where $b_{i}\in \mathbb{R}^{n},
F_{i}\in\mathbb{R}^{n\times n}$ for each $i\in V$.  In general, the
optimization problem~\eqref{PB:max_tr} is NP-hard, as we justify
below.  The next result shows that the function mapping $S$ to
$\mathcal{W}(S)$ displays the increasing returns property.

\begin{theorem}\longthmtitle{Increasing returns property of the
    function mapping $S$ to $\mathcal{W}(S)$}\label{Theo7}
  For any $S_{1}\subseteq S_{2}\subseteq V$ and $s\in V\backslash
  S_{2}$,
  \begin{equation}
    \mathcal{W}(S_{2}\cup \{s\})-\mathcal{W}(S_{2})\geq \mathcal{W}(S_{1}\cup
    \{s\})-\mathcal{W}(S_{1}).  \label{eq:Theo7}
  \end{equation}
\end{theorem}
\begin{proof}
  Without loss of generality, we relabel the elements in $V$ such that
  $%
  S_{1}=\{1,\cdots ,m_{1}\}$, $S_{2}=\{1,\cdots ,m_{1}+m_{2}\}$ and $%
  s=m_{1}+m_{2}+1$ with $m_{1}\geq 1$, $m_{2}\geq 0$ and
  $m_{1}+m_{2}+1\leq |V|=M$. For any $S=\{s_{1},\cdots
  ,s_{m}\}\subseteq V$, we have%
  \begin{align}
    \mathcal{W}_{1}(S)& =\sum_{k=0}^{\infty
    }A^{k}B_{S}B_{S}^{T}(A^{T})^{k}
    \notag \\
    & =\sum_{j=1}^{m}\sum_{k=0}^{\infty
    }A^{k}b_{s_{j}}b_{s_{j}}^{T}(A^{T})^{k} =\sum_{s\in
      S}\mathcal{W}_{1}(s), \label{eq1:PfTheo6}
  \end{align}
  which implies that
  \begin{align}
    \mathcal{W}_{1}(S_{2}\cup \{s\})-\mathcal{W}_{1}(S_{2})
    &=\mathcal{W}_{1}(s) \notag
    \\
    &=\mathcal{W}_{1}(S_{1}\cup \{s\})-\mathcal{W}_{1}(S_{1}).
    \label{eq1:PfTheo7}
  \end{align}%
  For $i\geq 2$, using the last equation in~\eqref{eq:W_i_W_i-1} and
  cancelling terms in common, we obtain
  \begin{align*}
    & \mathcal{W}_{i}(S_{2}\cup
    \{s\})-\mathcal{W}_{i}(S_{2})-\mathcal{W}%
    _{i}(S_{1}\cup \{s\})+\mathcal{W}_{i}(S_{1}) \\
    & \quad =\sum_{j_{1},j_{2}=1}^{m_{1}+m_{2}+1}\sum_{k=0}^{\infty
    }A^{k}F_{j_{2}}\mathcal{W}_{i-1}(s_{j_{1}})F_{j_{2}}^{T}(A^{T})^{k} \\
    & \quad \quad
    -\sum_{j_{1},j_{2}=1}^{m_{1}+m_{2}}\sum_{k=0}^{\infty
    }A^{k}F_{j_{2}}\mathcal{W}_{i-1}(s_{j_{1}})F_{j_{2}}^{T}(A^{T})^{k} \\
    & \quad \quad -\sum_{j_{1},j_{2}\in S_{1}\cup
      \{s\}}\sum_{k=0}^{\infty
    }A^{k}F_{j_{2}}\mathcal{W}_{i-1}(s_{j_{1}})F_{j_{2}}^{T}(A^{T})^{k} \\
    & \quad \quad +\sum_{j_{1},j_{2}=1}^{m_{1}}\sum_{k=0}^{\infty
    }A^{k}F_{j_{2}}%
    \mathcal{W}_{i-1}(s_{j_{1}})F_{j_{2}}^{T}(A^{T})^{k} \\
    & \quad =\sum_{j=m_{1}+1}^{m_{2}}\sum_{k=0}^{\infty
    }A^{k}F_{m_{1}+m_{2}+1}%
    \mathcal{W}_{i-1}(s_{j})F_{m_{1}+m_{2}+1}^{T}(A^{T})^{k} \\
    & \quad \quad +\sum_{j=m_{1}+1}^{m_{2}}\sum_{k=0}^{\infty
    }A^{k}F_{j}%
    \mathcal{W}_{i-1}(s_{m_{1}+m_{2}+1})F_{j}^{T}(A^{T})^{k} \\
    & \quad \geq 0.
  \end{align*}%
  The proof then follows using the definition~\eqref{eq:reach_Gramian}
  of $%
  \mathcal{W}$.
\end{proof}


Since the trace function is linear, Theorem~\ref{Theo7} immediately
implies that $\func{tr}(\mathcal{W})$\ is a supermodular function.
For linear time-invariant systems, the inequality in (\ref{eq:Theo7})
becomes an equality. This can be seen from the proof of
Theorem~\ref{Theo7} or found in~\cite{THS-JL:14}. Supermodularity in
combinatorial optimization of functions over subsets is analogous to
convexity in optimization of functions defined over Euclidean
space. The maximization of supermodular functions under cardinality
constraints is known to be NP-hard, but its Lagrangian dual and its
continuous relaxation can be solved in polynomial
time~\cite{GG-BS:89}, which provides an upper bound on the optimal
value.  On the other hand, a lower bound on the optimal value follows
from the following result.

\begin{corollary}
For any $S_{1}\subseteq S_{2}\subseteq V$ and $S_{3}\subseteq V\backslash
S_{2}$,
\begin{equation}
\mathcal{W}(S_{2}\cup S_{3})-\mathcal{W}(S_{2})\geq \mathcal{W}(S_{1}\cup
S_{3})-\mathcal{W}(S_{1}).  \label{eq:Cor1}
\end{equation}
\end{corollary}

\begin{proof}
Denote $S_{3}=\{s_{1},\cdots ,s_{\left\vert S_{3}\right\vert }\}$, it
follows immediately from~\eqref{eq:Theo7} that for any $i=1,\ldots
,\left\vert S_{3}\right\vert $,%
\begin{align}
& \mathcal{W}(S_{2}\cup \{s_{1},\cdots ,s_{i}\})-\mathcal{W}(S_{2}\cup
\{s_{1},\cdots ,s_{i-1}\}) \notag \\
& \geq \mathcal{W}(S_{1}\cup \{s_{1},\cdots
,s_{i}\})-\mathcal{W}(S_{1}\cup \{s_{1},\cdots ,s_{i-1}\}).
\label{eq1:Pf_Cor1}
\end{align}%
By summing inequality \eqref{eq1:Pf_Cor1} with respect to $i=1,\ldots
,\left\vert S_{3}\right\vert $, we obtain \eqref{eq:Cor1}.
\end{proof}

\begin{theorem}
\longthmtitle{Lower bound on reachability metrics}\label{Theo9}
  Let $f:\mathbb{R}^{n\times n}\rightarrow \mathbb{R}_{\geq 0}$ be either $\func{tr}$,
 $\lambda _{\func{min}}$ or $\det $. Then for any set $S$ of $m$ actuators
  \begin{equation}
    f(\mathcal{W}(S))\geq \sum_{i=1}^{N}f(\mathcal{W}(S_{i})),  \label{eq:Theo9}
  \end{equation}%
where $S_{1},\cdots ,S_{N}$ is any partition of $S$.
\end{theorem}

\begin{proof}
By letting $S_{1}=\emptyset $ in~\eqref{eq:Cor1} and using the fact that $%
\mathcal{W}(\emptyset )=\mathbf{0}_{n\times n}$, it holds immediately that
for any subset $S\subseteq V$, $\mathcal{W}(S)\geq \sum_{i=1}^{N}\mathcal{W}%
(S_{i})$. This directly implies that $\func{tr}(\mathcal{W}(S))\geq
\sum_{i=1}^{N}\func{tr}(\mathcal{W}(S_{i}))$. Moreover,
  \begin{align*}
    \lambda _{\func{min}}(\mathcal{W}(S))& =\min_{\lVert x\lVert =1}x^{T}%
    \mathcal{W}(S)x \\
    & \geq \min_{\lVert x\lVert =1}x^{T}\bigl(\sum_{i=1}^{N}\mathcal{W}(S_{i})%
    \bigr)x \\
    & \geq \sum_{i=1}^{N}\min_{\lVert x\lVert =1}x^{T}\mathcal{W}%
    (S_{i})x=\sum_{i=1}^{N}\lambda _{\func{min}}(\mathcal{W}(S_{i})).
  \end{align*}%
Finally, employing the Minkowski's determinant inequality~\cite{RAH-CRJ:85}
(if $A$, $B\in \mathbb{R}^{n\times n}$ are positive definite matrices, then $%
\det (A+B)\geq \det (A)+\det (B)$) repeatedly, we obtain
  \begin{equation*}
    \det (\mathcal{W}(S))\geq \sum_{i=1}^{N}\det (\mathcal{W}(S_{i}))
  \end{equation*}%
and the proof is complete.
\end{proof}


To maximize the lower bound in~\eqref{eq:Theo9}, one simply needs to
compute $f(\mathcal{W}(s))$ individually for every $s\in V$, order the
results in decreasing order, and select the actuators sequentially
starting with the one with largest value.  We refer to this procedure
as the \emph{greedy algorithm}.  The following example illustrates its
performance.

\begin{example}\longthmtitle{Controller selection via the greedy
    algorithm}\label{Example4}
  Consider an augmented bilinear control system based on the model in
  Example~\ref{Example3},
  \begin{equation*}
    x(k+1) = A x(k)+\sum_{j\in S}(F_{j}x(k)+B_{j})u_{j}(k),
  \end{equation*}
  where $A$, $B_{0}=b$ and $F_{0}=f$ are the same as those given in
  (\ref{eq1:Ex3}).
  The other actuator candidates are $(F_{j},B_{j})$, where $B_{j}$ is
  the $j$-th canonical vector in $\mathbb{R}^{5}$ for $j \in
  \{1,2,3\}$. We let $ F_{1}(1,2)=F_{1}(2,3)=0.02$, $F_{2}(2,5)=0.01$,
  $F_{2}(4,2)=0.05$,
  and $ F_{3}(1,1)=0.05$, $F_{3}(4,5)=0.02$, with all the other
  elements in $F_{j}$ being zero, for $j \in
  \{1,2,3\}$. Table~\ref{table1} shows their individual and combined
  contributions to $\lambda _{\func{min}}(\mathcal{W}(S))$, $
  \func{tr}(\mathcal{W}(S))$, and $\det (\mathcal{W}(S))$.
  \begin{table*}[tbh]
    \caption{Contribution of chosen sets of actuators to the Gramian-based
      reachability metrics.}\label{table1}
    \centering%
    \begin{tabular}{||c|c|c|c||c|c|c|c||}
      \hline
      $S$ & $\func{tr}(W(S))$ & $\lambda _{\func{min}}(W(S))$ & $\det
      (W(S))$ & $S$
      & $\func{tr}(W(S))$ & $\lambda _{\func{min}}(W(S))$ & $\det
      (W(S))$
      \\
      \hline
      $\{0\}$ & $14.42$ & $0.027$ & $0.242$ & $\{0,2\}$ & $19.91$ &
      $0.07$ & $3.32$
      \\
      \hline
      $\{1\}$ & $5.03$ & $0.023$ & $0.025$ & $\{0,3\}$ & $18.69$ &
      $0.05$ & $1.13$
      \\ \hline
      $\{2\}$ & $4.04$ & $3\times 10^{-5}$ & $9\times 10^{-7}$ & $\{0,1,2\}$ & $%
      26.50$ & $0.137$ & $46.15$
      \\
      \hline
      $\{3\}$ & $3.03$ & $1.6\times 10^{-6}$ & $4\times 10^{-11}$ &
      $\{0,1,3\}$ & $%
      25.28$ & $0.125$ & $28.68$
      \\
      \hline
      $\{0,1\}$ & $20.98$ & $0.09$ & $11.704$ & $\{0,2,3\}$ & $24.19$
      & $0.103$ & $
      8.34$
      \\
      \hline
    \end{tabular}%
  \end{table*}
  We make the following observations:
  \begin{enumerate}
  \item Actuators with a large individual contribution provide a large
    combinatorial contribution. This fact suggests that the greedy
    algorithm is a sensible strategy, even though $\sum_{s\in
      S}f(\mathcal{W}(s))$ can be considerably smaller than
    $f(\mathcal{W}(S))$ for $f=\lambda _{\func{min}}$ and $\det$.

  \item For $f = \func{tr}$, $\sum_{s\in S}\func{tr}(\mathcal{W}(s))$
    is a good estimate of $\func{%
      tr}(\mathcal{W}(S))$. For example,%
    \begin{align*}
      & \frac{\func{tr}(\mathcal{W}(\{0,1\}))-\sum_{s\in
          \{0,1\}}\func{tr}(%
        \mathcal{W}(s))}{\func{tr}(\mathcal{W}(\{0,1\}))} \\
      & \quad =\frac{20.98-14.42-5.03}{20.98} =0.073,
    \end{align*}%
    i.e., a relative error less than $8\%$.
  \end{enumerate}
  As an example, for the case $m=2$, the greedy algorithm will
  select~$\{0,1\}$, which is the optimal choice.  We have observed
  similar results for various simulations of this example with several
  sets of randomly generated $(B_{j},F_{j})$. \oprocend
\end{example}

\section{Addition of bilinear inputs to linear symmetric
  networks}\label{Sec:imp_bl}

In this section, we examine the effect that the addition of bilinear
inputs has on the worst-case minimum input energy for
difficult-to-control linear networks. We begin by formalizing this
notion.

\begin{definition}\longthmtitle{Difficult-to-control
    networks}\label{Def:dtc_network}
  A class of networks is said to be difficult to control (DTC) if, for a
  fixed number of control inputs, the normalized worst-case minimum
  energy grows unbounded with the scale of the network, i.e.,
  \begin{align*}
\lim_{n\rightarrow \infty }\sup_{x_{f}\in \mathbb{R}^{n}}\inf_{\substack{ %
\{u\}^{\infty }:u(k) \in \mathbb{R}^{m}, \\ x(\infty )=x_{f}}}\frac{\lVert
\{u\}^{\infty }\rVert ^{2}}{\lVert x_{f}\rVert ^{2}}\rightarrow \infty ,
  \end{align*}
where $x(\infty) \triangleq \lim_{k\rightarrow \infty}x(k)$.
\end{definition}
\smallskip

For linear networks $(A(n),\mathbf{0}_{n\times nm},B(n))$, one can see
from~\eqref{eq:LTI_Gra_LB} that
\begin{equation*}
  \sup_{x_{f}\in \mathbb{R}^{n}:\lVert x_{f}\rVert ^{2}=1}\inf_{\substack{ %
\{u\}^{\infty }:u(k) \in \mathbb{R}^{m}, \\ x(\infty )=x_{f}}}\lVert \{u\}^{\infty }\rVert ^{2}=\lambda _{\min
  }^{-1}(\mathcal{W}_{1}(n)) .
\end{equation*}
Therefore, if the linear network is difficult to control, this implies
that the minimum eigenvalue of the reachability Gramian approaches $0$
as $n$ grows. A typical class of difficult-to-control linear networks
is the class of stable and symmetric networks for which,
cf.~\cite[Corollary 3.2]{FP-SZ-FB:14}, the worst-case minimum input
energy grows exponentially with rate $\frac{n}{m}$ for any choice of
$B(n)\in \mathbb{R}^{n \times m}$ whose columns are canonical vectors
in $\mathbb{R}^{n}$ (here, $m$ is the number of control inputs or
control nodes).


Our first result of this section shows that difficult-to-control
linear symmetric networks remain so after the addition of a finite number of
bilinear inputs.

\begin{theorem}\longthmtitle{Difficult-to-control linear symmetric networks
    remain so after granted the ability to control a finite number of
    interconnections}\label{Theo6}
  Consider a class of difficult-to-control linear symmetric networks
  $(A(n),\mathbf{0}_{n\times nm},B(n))$.  The class of bilinear
  networks $(A(n),F(n),B(n))$ is also DTC if the number of nonzero
  entries in the matrix $F(n)\in \mathbb{R}^{n\times nm}$ and $\rVert
  F(n)\rVert _{\max }\triangleq \max_{i,j}\left\vert
    F_{ij}(n)\right\vert $ are uniformly bounded with respect to $n$.
\end{theorem}
\begin{proof}
  Our proof has two parts. First, we construct a class of bilinear
  control system whose trajectories include the trajectories of the
  systems $(A(n),F(n),B(n))$. Second, we establish a correspondence
  between the constructed bilinear systems and linear control
  networks, and build on it to show that they are difficult to
  control. For the first step, let $\lvert \cdot \rvert_{nz}$ denote
  the number of nonzero entries in a matrix and define $
  M_{F}=\sum_{j=1}^{m}\lvert F_{j}(n)\rvert _{nz}$. Select matrices
  $\hat{F} _{j}(n) $ with $\rvert \hat{F}_{j}(n)\rvert _{nz}=1$ for $j
  \in \until{M_{F}}$ and
  \begin{align*}
    \sum_{j=\lvert F_{i-1}(n)\rvert _{nz} +1}^{\lvert F_{i}(n)\rvert
      _{nz}}\hat{F}_{j}(n)=F_{i}(n) ,
  \end{align*}
  for $i \in \until{m}$, where $\lvert F_{0}(n)\rvert _{nz}\triangleq
  0$ for convenience.  Consider the bilinear system
  \begin{equation}
    x(k+1)=A(n)x(k)+\sum_{j=1}^{M_{F}}\hat{F}_{j}(n)x(k)v_{j}(k)+%
    \sum_{j=1}^{m}B_{j}(n)u_{j}(k),  \label{eq1:Pf_Theo6}
  \end{equation}%
  with state $x(k)\in \mathbb{R}^{n}$, inputs $u_{j}(k),v_{j}(k)\in
  \mathbb{R}$, and system matrices $A(n),\hat{F}_{j}(n)\in \mathbb{R}%
  ^{n\times n}$, $B_{j}(n)\in \mathbb{R}^{n}$. Note that, selecting
  $v_{j}(k) =u_{i}(k)$ for $j=\lvert F_{i-1}(n)\rvert _{nz}+1,\ldots
  ,\lvert F_{i}(n)\rvert _{nz}$ and $i \in \until{m}$
  makes~\eqref{eq1:Pf_Theo6} take the form
  \begin{equation*}
    x(k+1)=A(n)x(k) + \sum_{j=1}^{m} (F_{j}(n)x(k) +
    B_{j}(n))u_{j}(k), 
  \end{equation*}
  which corresponds to the bilinear network $ (A(n),F(n),B(n))$. This
  implies that the optimal control sequence $\{u^{\ast }\}^{\infty }$
  for $(A(n),F(n),B(n))$ generates the same network state trajectory
  as the (not necessarily optimal) control sequence $\{v^{\ast
  },u^{\ast }\}^{\infty }$ for~\eqref{eq1:Pf_Theo6}.

  Our second step establishes that the bilinear network
  in~\eqref{eq1:Pf_Theo6} is difficult to control. Assume that the
  nonzero entry in $\hat{F}_{j}(n)$ is in the $i$-th row. Then, there
  exist $M_{F}$ scalar sequences $\{\hat{u}_{j}(k)\}$ such that for
  all $j \in \until{M_{F}}$ and all $k\in \mathbb{Z}_{\geq 0}$,
  \begin{equation}
    \hat{F}_{j}(n)x(k)v_{j}(k)=\hat{B}_{j}\hat{u}_{j}(k),  \label{eq3:Pf_Theo6}
  \end{equation}%
  where $\hat{B}_{j}=e_{i}$ is the $i$-th canonical unit vector in
  $\mathbb{R}^{n}$. Substituting this into~\eqref{eq1:Pf_Theo6}, we
  obtain
  \begin{equation}
    x(k+1)=A(n)x(k)+\sum_{j=1}^{M_{F}}\hat{B}_{j}\hat{u}_{j}(k)+%
    \sum_{j=1}^{m}B_{j}(n)u_{j}(k),  \label{eq4:Pf_Theo6}
  \end{equation}%
  which is linear, symmetric and difficult to control because its
  number of control nodes is at most $m+M_{F}$, which is
  constant. Furthermore, there exists a constant $\bar{X}\in
  \mathbb{R}_{>0}$, such that for all $i=1,\ldots ,n$, and all $k\in
  \mathbb{Z}_{\geq 0}$, $\rVert x_{i}(k)\rVert \leq \bar{X}$ since the
  state trajectory starts at the origin, the network is stable and the
  input is bounded. As a result,%
  \begin{align*}
    & \sup_{x_{f}\in \mathbb{R}^{n}}\inf_{\underset{\text{equation
          (\ref%
          {eq1:Pf_Theo6})}}{\{u,v\}^{\infty }}}\lVert \{u,v\}^{\infty
    }\rVert
    ^{2}\lVert x_{f}\rVert ^{-2} \\
    & \quad \overset{\text{(\ref{eq3:Pf_Theo6})}}{\geq }\sup_{x_{f}\in
      \mathbb{R}%
      ^{n}:\lVert x_{f}\rVert ^{2}=1}\inf_{\underset{\text{equation
          (\ref%
          {eq4:Pf_Theo6})}}{\{u,\hat{u}\}^{\infty }}}(\lVert
    \{u\}^{\infty }\rVert ^{2}+\rVert F\rVert _{\max
    }^{-2}\bar{X}^{-2}\lVert \{\hat{u}\}^{\infty
    }\rVert ^{2}) \\
    & \quad \geq \min (1,\rVert F\rVert _{\max }^{-2}\bar{X}^{-2}) \\
    & \quad \quad \cdot \sup_{x_{f}\in \mathbb{R}^{n}:\lVert
      x_{f}\rVert ^{2}=1}\inf_{\underset{\text{equation
          (\ref{eq4:Pf_Theo6})}}{\{u,\hat{u}%
        \}^{\infty }}}(\lVert \{u\}^{\infty }\rVert ^{2}+\lVert
    \{\hat{u}\}^{\infty }\rVert ^{2}),
  \end{align*}
  which implies the result.
\end{proof}

Our next result shows that difficult-to-control linear symmetric
networks might remain so even after relaxing the finiteness of the
number of interconnections that can be affected by the addition of
bilinear inputs. More concretely, we study the reachability properties
of the class of networks $(A,F,B)$ with symmetric adjacency matrices
$A=A^{T}$ and $F=\alpha I_{n}$, $|\func{tr} (F)| \leq \mu(n) \in o
(\sqrt{n})$ (without loss of generality, we let $\alpha \geq 0$). For
instance, this corresponds to the case when a central controller can
affect the strengths of the self-loops of all agents simultaneously in
a linear symmetric network or when all agents simultaneously adjust
the strength of their self-loops by the same amount.

\begin{theorem}\longthmtitle{Worst-case control energy for linear symmetric
    networks with self-loop modulation}\label{Theo8}
  Consider the class of bilinear networks given by
  \begin{align*}
    x(k+1) = (A+\alpha v(k)I_{n})x(k) + \sum_{j=1}^{m}B_{j}u_{j}(k) ,
  \end{align*}%
  with $A=A^{T}$, $|%
  \func{tr}(\alpha I_{n})|\leq \mu(n) \in o(\sqrt{n})$ and $\rho
  (A)<\sqrt{1-T_{m}^{-1}}$, where $ T_{m}\triangleq \left\lceil
    \frac{n}{m}\right\rceil -1$. Then the reachability Gramian of the
  network satisfies, for any $ n>m^{-1}\mu^{2}(n)$,
  \begin{equation}
    \lambda _{\min }(\mathcal{W})\leq \frac{(1-T_{m}\alpha ^{2})^{-1}}{1-\rho
      ^{2}(A)-T_{m}^{-1}}\rho ^{2T_{m}}(A).  \label{eq:Thm_8}
  \end{equation}
\end{theorem}
\begin{proof}
  Define $\func{sum}(k,i)\triangleq k_{1}+\cdots +k_{i}$ for $ k\in
  \{\mathbb{Z}_{\geq 0}\}^{\infty}$ and $i\in \mathbb{Z}_{\geq 1}$.
  By definition of the reachability Gramian (\ref{eq:reach_Gramian}), it
  follows that
  \begin{align*}
    \mathcal{W}_{i}& =\alpha ^{2(i-1)}\sum_{k_{1},\cdots ,k_{i}=0}^{\infty }A^{
      \func{sum}(k,i)}BB^{T}(A^{T})^{\func{sum}(k,i)} \\
    & =\mathcal{W}_{i,s}+\mathcal{W}_{i,f},
  \end{align*}
  where
  \begin{align*}
    \mathcal{W}_{i,s}& \triangleq \alpha ^{2(i-1)}\sum_{k_{1},\cdots ,k_{i}=0}^{
      \func{sum}(k,i)<T_{m}}A^{\func{sum}(k,i)}BB^{T}(A^{T})^{\func{sum}(k,i)}, \\
    \mathcal{W}_{i,f}& \triangleq \alpha ^{2(i-1)}\sum_{k_{1},\cdots ,k_{i}=0}^{
      \func{sum}(k,i)\geq T_{m}}A^{\func{sum}(k,i)}BB^{T}(A^{T})^{\func{sum}(k,i)}.
  \end{align*}
  Therefore,
  \begin{align}  \label{eq:lamda_min_self_loop}
    \lambda _{\min }(\mathcal{W})& =\lambda_{\min } \big( \sum_{i=1}^{\infty }
    \mathcal{W}_{i,s}+\sum_{i=1}^{\infty }\mathcal{W}_{i,f} \big) \\
    & \leq \lambda _{\min } \big( \sum_{i=1}^{\infty }\mathcal{W}_{i,s}\big) +
    \sum_{i=1}^{ \infty }\lVert \mathcal{W}_{i,f}\rVert \leq \sum_{i=1}^{\infty
    }\lVert \mathcal{W}_{i,f}\rVert ,  \notag
  \end{align}
  where the first inequality follows from the Bauer-Fike theorem~\cite%
  {RAH-CRJ:85} and the second inequality follows by noting that $%
  \sum_{i=1}^{\infty }\mathcal{W} _{i,s}$ is singular because its column space
  is contained in the space spanned by $\{B,AB,\cdots ,A^{T_{m}-1}B\}$, whose
  dimension is smaller than $n$ by definition of $T_{m}$. We can write $%
  \mathcal{W}_{i,f}$ in a recursive manner as follows,
  \begin{align*}
    \mathcal{W}_{i,f} & =\alpha ^{2(i-1)}\sum_{k_{i}=0}^{\infty
    }A^{k_{i}}\bigl(%
    \sum_{k_{1},\cdots ,k_{i-1}=0}^{\func{sum}(k,i-1)\geq
      T_{m}}A^{\func{sum}%
      (k,i-1)}B \\
    & \quad \quad \cdot
    B^{T}(A^{T})^{\func{sum}(k,i-1)}\bigr)(A^{T})^{k_{i}}
    \\
    & \quad +\alpha
    ^{2(i-1)}\sum_{j=0}^{T_{m}-1}\sum_{k_{i}=T_{m}-j}^{\infty } %
    \bigl(\sum_{k_{1},\cdots
      ,k_{i-1}=0}^{\func{sum}(k,i-1)=j}A^{\func{sum}
      (k,i)}B \\
    & \quad \quad \cdot B^{T}(A^{T})^{\func{sum}(k,i)}\bigr) \\
    & = \alpha ^{2}\sum_{k_{i}=0}^{\infty }A^{k_{i}}\mathcal{W}%
    _{i-1,f}(A^{T})^{k_{i}} \\
    & \quad + \alpha^{2(i-1)} \displaystyle \sum_{j=0}^{T_{m}-1}
    \mathcal{\eta}%
    (i-1,j)\mathcal{W}_{1,f} ,
  \end{align*}
  where $\mathcal{\eta }(N,M)$ is the number of ways of choosing $N
  \in \mathbb{Z}_{\geq 0}$ non-negative integers such that their sum
  equals $M \in \mathbb{Z}_{\geq 0}$. Two properties of this function
  are useful to us: (i) $%
  \mathcal{\eta }(N,M)=\sum_{j=0}^{M}\mathcal{\eta }(N-1,j)$ and (ii)
  $%
  \mathcal{\eta }(N,M)$ is an increasing function of $N$ and
  $M$. Using (i), we obtain 
  \begin{align*}
    \mathcal{W}_{i,f} = \alpha^{2} \! \displaystyle
    \sum_{k_{i}=0}^{\infty
    }A^{k_{i}}\mathcal{W}_{i-1,f}(A^{T})^{k_{i}} \!+\! \alpha
    ^{2(i-1)}\mathcal{%
      \eta }(i,T_{m}-1)\mathcal{W}_{1,f}.
  \end{align*}
  Taking norms and upper bounding, we get
  \begin{align*}
    \lVert \mathcal{W}_{i,f}\rVert & \leq \frac{\alpha ^{2}}{ 1-\lVert A\rVert
      ^{2}}\lVert \mathcal{W}_{i-1,f}\rVert +\alpha ^{2(i-1)} \mathcal{\eta }%
    (i,T_{m}-1)\lVert \mathcal{W}_{1,f}\rVert .
  \end{align*}
  Using this inequality repeatedly, we obtain
  \begin{align*}
    \lVert \mathcal{W}_{i,f}\rVert & \leq \sum_{j=0}^{i-1}\frac{\alpha
      ^{2(i-1)}%
    }{(1-\lVert A\rVert ^{2})^{j}}\mathcal{\eta }(i-j,T_{m}-1)\lVert
    \mathcal{W}%
    _{1,f}\rVert \\
    & \leq (T_{m}\alpha ^{2})^{(i-1)}\lVert \mathcal{W}_{1,f}\rVert
    \sum_{j=0}^{i-1}T_{m}^{-j}(1-\lVert A\rVert ^{2})^{-j}
  \end{align*}
  where we have used $\mathcal{\eta }(N,M) \leq (M+1)\mathcal{\eta
  }(N-1,M) \leq (M+1)^{N-1}\mathcal{\eta }(1,M) =(M+1)^{N-1}$, which
  follows from properties (i) and (ii) of $\eta$ above. Since $A$\ is
  symmetric and Schur stable, $\lVert A\rVert =\rho (A)$, which
  together with $\rho(A)<\sqrt{%
    1-T_{m}^{-1}}$ implies $T_{m}^{-1}(1-\lVert A\rVert ^{2})^{-1}<1$.
  Therefore, we conclude
  \begin{align}
    \lVert \mathcal{W}_{i,f}\rVert & \leq \frac{(T_{m}\alpha ^{2})^{(i-1)}\lVert
      \mathcal{W}_{1,f}\rVert }{ 1-T_{m}^{-1}(1-\lVert A\rVert ^{2})^{-1}} .
    \label{eq:W_if_norm}
  \end{align}
  Combining (\ref{eq:lamda_min_self_loop}) with (\ref{eq:W_if_norm}), we
  obtain
  \begin{align*}
    \lambda _{\min }(\mathcal{W})& \leq \sum_{i=1}^{\infty }\frac{(T_{m}\alpha
      ^{2})^{(i-1)}\lVert \mathcal{W}_{1,f}\rVert }{1-T_{m}^{-1}(1-\rho
      ^{2}(A))^{-1}} \\
    & =\frac{(1-T_{m}\alpha ^{2})^{-1}}{1-T_{m}^{-1}(1-\rho ^{2}(A))^{-1}}\lVert
    \mathcal{W}_{1,f}\rVert ,
  \end{align*}
  where we have used the fact that $|\func{tr}(F)|\leq \mu(n)$ implies
  that $%
  T_{m}\alpha ^{2}<1$ for $n>m^{-1}\mu^{2}(n)$.  Using \cite[Theorem
  3.1]{FP-SZ-FB:14}, we obtain
  \begin{align*}
    \lambda _{\min }(\mathcal{W}) &\leq \frac{(1-T_{m}\alpha
      ^{2})^{-1}}{ 1-T_{m}^{-1}(1-\rho ^{2}(A))^{-1}}\frac{\rho
      ^{2T_{m}}(A)}{1-\rho ^{2}(A)}
    \\
    &=\frac{(1-T_{m}\alpha ^{2})^{-1}}{1-\rho ^{2}(A)-T_{m}^{-1}}\rho
    ^{2T_{m}}(A),
  \end{align*}
  and the proof is complete.
\end{proof}

Note that, for a large-scale network with a fixed number of control
nodes, the assumption that $\rho (A)<\sqrt{1-T_{m}^{-1}}$ in
Theorem~\ref{Theo8} is not restrictive because $\sqrt{ 1-T_{m}^{-1}}$
becomes arbitrarily close to $ 1$ as $n$ increases.  One can show that
$\frac{%
  (1-T_{m}\alpha ^{2})^{-1}}{1-\rho ^{2}(A)-T_{m}^{-1}}$
in~\eqref{eq:Thm_8} is a decreasing function of $n$ and that
\begin{equation*}
  \lim_{n\rightarrow \infty }\frac{(1-T_{m}\alpha ^{2})^{-1}}{1-\rho
    ^{2}(A)-T_{m}^{-1}}=(1-\rho ^{2}(A))^{-1}.
\end{equation*}%
Thus, $\lambda _{\min }(\mathcal{W})$ decreases at least exponentially
as $n$ increases, which means the worst-case control energy increases
exponentially, as indicated by Theorem~%
\ref{Theo5}.  Therefore Theorem~\ref{Theo8} can be interpreted as
saying that bounded homogeneous self-loop modulation through bilinear
inputs does not make a linear symmetric network easier to control.


We illustrate the result in Theorem~\ref{Theo8} with an example.
%
%
\begin{example}\longthmtitle{Line network with self-loop
    modulation}\label{Example5}
  Consider the group of line networks for $n \in \until{15}$ with
  adjacency matrices $A=[a_{ij}]$, where $%
  a_{ij}=0.25 $ if $|i-j|\leq 1$ and $a_{ij}=0$ otherwise for $i,j \in
  \QTR{until}{n}$. Let $m=3$, with $B_{1}, B_{2}, B_{3}$ being
  canonical vectors chosen optimally using exhaustive search to
  maximize $\lambda _{\min }(\mathcal{W})$, and let $|\func{tr}(\alpha
  I_n)|=0.9$. The minimum eigenvalue of the reachability Gramian is
  plotted in a logarithmic scale in Figure \ref%
  {Fig2} as a function of $n$. It can be seen that $\lambda _{\min
  }(\mathcal{W})$ decreases exponentially as $n$ increases, which
  implies that the worst-case control energy increases exponentially
  with the scale of the network, even with self-loop modulation.
  \begin{figure}[tbh]
    \centering
    \includegraphics[scale=0.5]{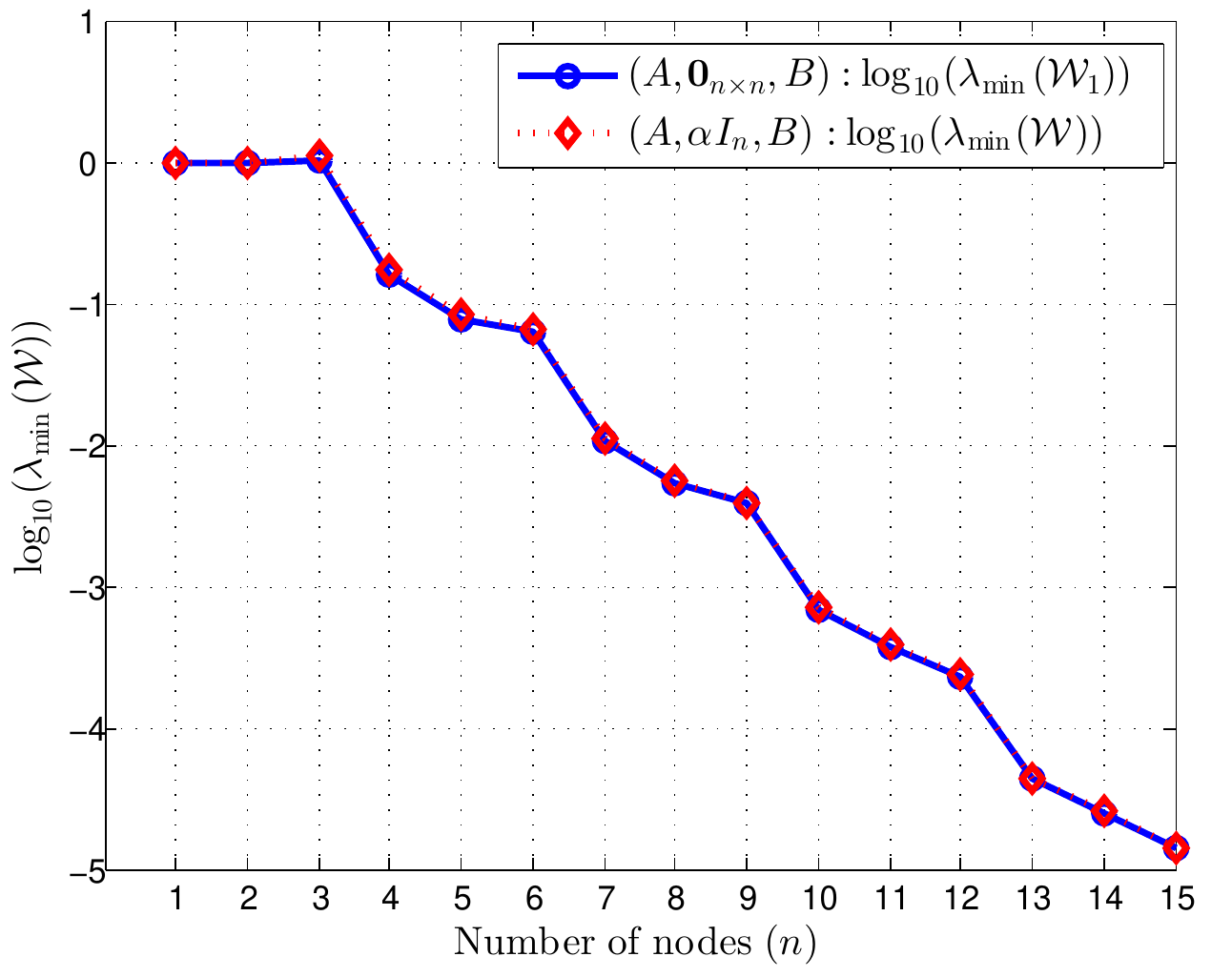}
    \caption{For the class of line networks described in Example
      \protect\ref%
      {Example5}, $\log _{10}(\protect\lambda_{\min}(\mathcal{W}))$ is
      plotted in red diamonds as the scale $n$ of the network
      increases from $1$ to $15$. The same quantity is also plotted in
      blue circles for the case without self-loop modulation
      ($F=\mathbf{0}_{n\times n}$). As predicted by Theorem
      \protect\ref%
      {Theo8}, symmetric networks with or without self-loop modulation
      are difficult to control with a fixed number of control
      nodes.}\label{Fig2}
  \end{figure}
  \relax\ifmmode\else\unskip\hfill\fi\hbox{$\bullet$}
\end{example}

We conclude this section with an example that shows that a
difficult-to-control linear network can be made easy to control by
adding a single bilinear input that affects an infinite number of
interconnections with strength that is independent of the scale of the
network.

\begin{example}\longthmtitle{Linear symmetric line network with $n$-dependent
    interconnection modulation}\label{Example7}
  Consider the group of bilinear networks $(A(n),F(n),B(n))$ with
  \begin{eqnarray*}
    A(n) &=&\left[
      \begin{array}{ccccc}
        0.05 & 0.05 & 0 & \cdots & 0 \\
        0.05 & 0.05 & 0.05 & \ddots & \vdots \\
        0 & 0.05 & \ddots & \ddots & 0 \\
        \vdots & \ddots & \ddots & \ddots & 0.05 \\
        0 & \cdots & 0 & 0.05 & 0.05%
      \end{array}%
    \right] \in \mathbb{R}^{n\times n}, \\
    B(n) &=&\left[
      \begin{array}{cccc}
        1 & 0 & \cdots & 0%
      \end{array}%
    \right] ^{T}\in \mathbb{R}^{n},
  \end{eqnarray*}%
  and $F(n)=[f_{ij}]$ with $f_{ij}=1$ if $i=j+1$ and all the other
  entries $0$.  Figure \ref{Fig4} compares $\lambda _{\min
  }(\mathcal{W}_{1}\mathcal{)}$ of the linear line network
  $(A(n),\mathbf{0}_{n\times nm},B(n))$ with $\lambda _{\min }(%
  \mathcal{W})$ of the bilinear network $(A(n),F(n),B(n))$. One can
  see that $%
  \lambda _{\min }(\mathcal{W}_{1}\mathcal{)}$ decreases exponentially
  as the scale $n$ of the network increases, which implies that the
  linear network is difficult to control. By employing the bilinear
  control through $F(n)$, $ \lambda _{\min }(\mathcal{W)}$ is kept
  constant as $n$ increases. Note that the number of interconnections
  we need to modulate increases with $n$.
  \begin{figure}[tbh]
    \centering
    \includegraphics[scale=0.5]{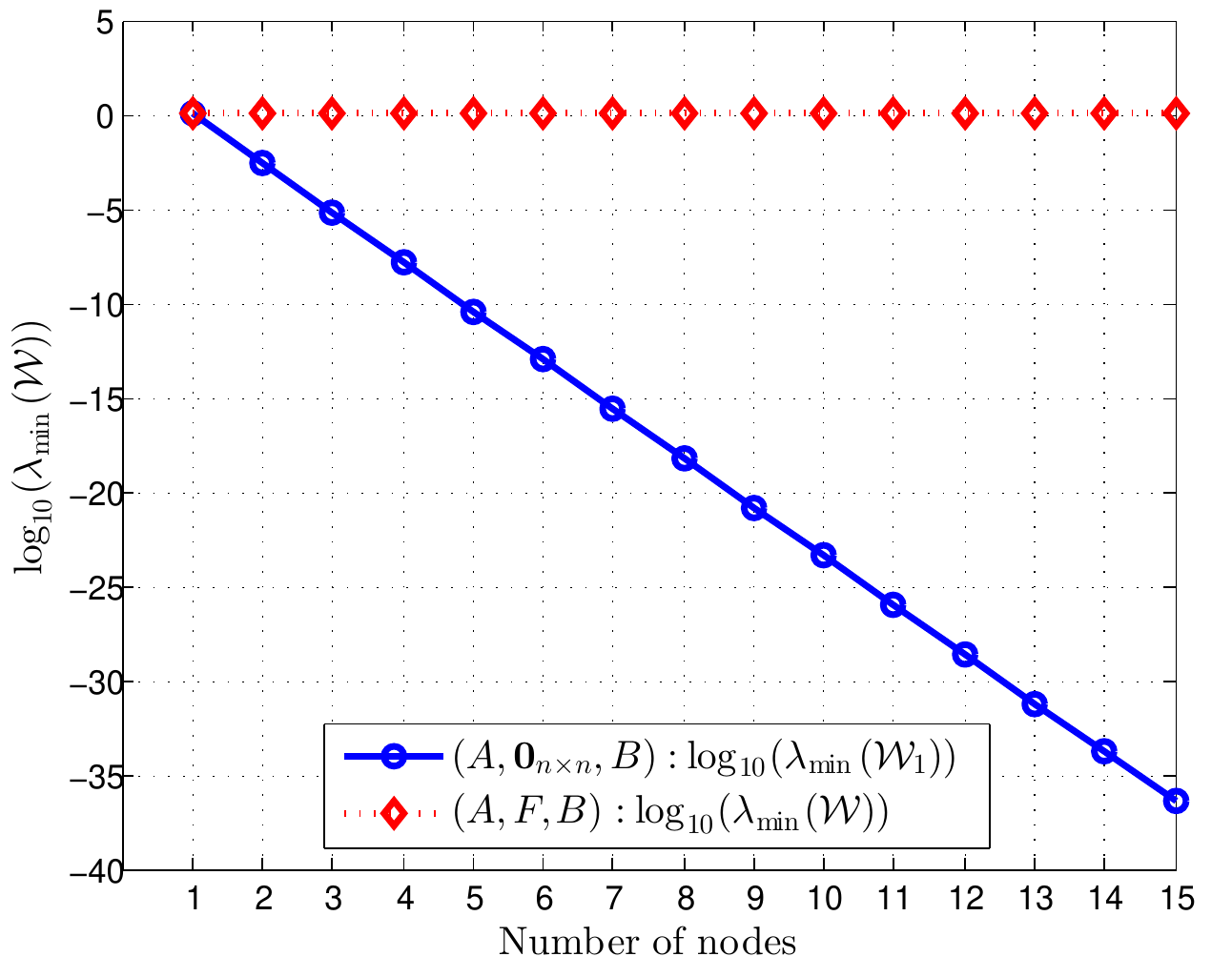}
    \caption{The class of linear networks $(A,\mathbf{0}_{n\times
        nm},B)$ are difficult to control while the corresponding
      bilinear networks $(A,F,B)$ are easy to control.}\label{Fig4}
    \vspace*{-1.5ex}
  \end{figure}
  \relax\ifmmode\else\unskip\hfill \fi\hbox{$\bullet$}
\end{example}

\section{Conclusions}\label{Sec:con}
We have proposed Gramian-based reachability metrics for discrete-time
bilinear control networks to quantify the input energy required to
steer the state from the origin to an arbitrary point. Our
reachability notions build on the fact that, when the infinity norm of
the input is upper bounded by some function of the system matrices,
then the required minimum input energy can be lower bounded in terms
of the reachability Gramian.  We have studied the supermodularity
properties of Gramian as a function of the actuators and derived lower
bounds on the reachability metrics in terms of the aggregate
contribution of the individual actuators.  Finally, we have studied
the effect that the addition of bilinear inputs has on the
difficult-to-control character of linear symmetric networks.  Future
work will include the design of algorithms for optimal selection of
control nodes in complex networks, where both the nodes and the
interconnection strength among neighboring nodes can be affected by
actuators, the study of the more general problem of steering the
network state from an arbitrary initial condition to an arbitrary
target state, and the analysis of observability metrics for bilinear
control systems based on the generalized observability Gramian.

\section*{Acknowledgments}
The authors would like to thank the anonymous reviewers for comments
that help improve the readability of the paper.  This work was
partially supported by NSF Award CNS-1329619.

\bibliographystyle{IEEEtran}%
\bibliography{alias,Main,Main-add,JC}

\begin{thebibliography}{10}
\providecommand{\url}[1]{#1}
\csname url@samestyle\endcsname
\providecommand{\newblock}{\relax}
\providecommand{\bibinfo}[2]{#2}
\providecommand{\BIBentrySTDinterwordspacing}{\spaceskip=0pt\relax}
\providecommand{\BIBentryALTinterwordstretchfactor}{4}
\providecommand{\BIBentryALTinterwordspacing}{\spaceskip=\fontdimen2\font plus
\BIBentryALTinterwordstretchfactor\fontdimen3\font minus
  \fontdimen4\font\relax}
\providecommand{\BIBforeignlanguage}[2]{{%
\expandafter\ifx\csname l@#1\endcsname\relax
\typeout{** WARNING: IEEEtran.bst: No hyphenation pattern has been}%
\typeout{** loaded for the language `#1'. Using the pattern for}%
\typeout{** the default language instead.}%
\else
\language=\csname l@#1\endcsname
\fi
#2}}
\providecommand{\BIBdecl}{\relax}
\BIBdecl

\bibitem{YZ-JC:15-cdc}
Y.~Zhao and J.~Cort\'{e}s, ``Reachability metrics for bilinear complex
  networks,'' in \emph{{IEEE} Conf.\ on Decision and Control}, Osaka, Japan,
  2015, pp. 4788--4793.

\bibitem{YYL-JJS-ALB:11}
Y.~Y. Liu, J.~J. Slotine, and A.~L. Barab{\'a}si, ``Controllability of complex
  networks,'' \emph{Nature}, vol. 473, no. 7346, pp. 167--173, 2011.

\bibitem{AO:14}
A.~Olshevsky, ``Minimal controllability problems,'' \emph{IEEE Transactions on
  Control of Network Systems}, vol.~1, no.~4, pp. 249--258, 2014.

\bibitem{AR-MJ-MM-ME:09}
A.~Rahmani, M.~Ji, M.~Mesbahi, and M.~Egerstedt, ``Controllability of
  multi-agent systems from a graph-theoretic perspective,'' \emph{SIAM Journal
  on Control and Optimization}, vol.~48, no.~1, pp. 162--186, 2009.

\bibitem{CA-BG:14}
C.~Aguilar and B.~Gharesifard, ``Necessary conditions for controllability of
  nonlinear networked control systems,'' in \emph{{A}merican {C}ontrol
  {C}onference}, Portland, OR, USA, 2014, pp. 5379--5383.

\bibitem{GY-JR-YL-CL-BL:12}
G.~Yan, J.~Ren, Y.~Lai, C.~Lai, and B.~Li, ``Controlling complex networks: How
  much energy is needed?'' \emph{Physical Review Letters}, vol. 108, no.~21, p.
  218703, 2012.

\bibitem{FP-SZ-FB:14}
F.~Pasqualetti, S.~Zampieri, and F.~Bullo, ``Controllability metrics,
  limitations and algorithms for complex networks,'' \emph{IEEE Transactions on
  Control of Network Systems}, vol.~1, no.~1, pp. 40--52, 2014.

\bibitem{THS-JL:14}
T.~Summers and J.~Lygeros, ``Optimal sensor and actuator placement in complex
  dynamical networks,'' in \emph{World Congress}, vol.~19, no.~1, 2014, pp.
  3784--3789.

\bibitem{VT-MAR-GJP-AJ:15}
V.~Tzoumas, M.~A. Rahimian, G.~J. Pappas, and A.~Jadbabaie, ``Minimal actuator
  placement with optimal control constraints,'' \emph{arXiv preprint
  arXiv:1503.04693}, 2015.

\bibitem{KJF-LH-WP:03}
K.~J. Friston, L.~Harrison, and W.~Penny, ``Dynamic causal modelling,''
  \emph{NeuroImage}, vol.~19, pp. 1273--1302, 2003.

\bibitem{JRI-AO-TM-MP-JS-GC-HP:14}
J.~R. Iversen, A.~Ojeda, T.~Mullen, M.~Plank, J.~Snider, G.~Cauwenberghs, and
  H.~Poizner, ``Causal analysis of cortical networks involved in reaching to
  spatial targets,'' in \emph{Annual Int. Conf. of the IEEE Engineering in
  Medicine and Biology Society}, Chicago, IL, 2014, pp. 4399--4402.

\bibitem{CB-GD-GK:74}
C.~Bruni, G.~Dipillo, and G.~Koch, ``Bilinear systems: An appealing class of
  "nearly linear" systems in theory and applications,'' \emph{IEEE Transactions
  on Automatic Control}, vol.~19, no.~4, pp. 334--348, 1974.

\bibitem{DE:09}
D.~Elliott, \emph{Bilinear Control Systems: Matrices in Action}.\hskip 1em plus
  0.5em minus 0.4em\relax Springer Science \& Business Media, 2009, vol. 169.

\bibitem{PP-VY:10}
P.~Pardalos and V.~Yatsenko, \emph{Optimization and Control of Bilinear
  Systems: Theory, Algorithms, and Applications}.\hskip 1em plus 0.5em minus
  0.4em\relax Springer Science \& Business Media, 2010, vol.~11.

\bibitem{DK-KN:85}
D.~Koditschek and K.~Narendra, ``The controllability of planar bilinear
  systems,'' \emph{IEEE Transactions on Automatic Control}, vol.~30, no.~1, pp.
  87--89, 1985.

\bibitem{UP-PF:92}
U.~Piechottka and P.~Frank, ``Controllability of bilinear systems,''
  \emph{Automatica}, vol.~28, no.~5, pp. 1043--1045, 1992.

\bibitem{TG-TT-JZ:73}
T.~Goka, T.~Tarn, and J.~Zaborszky, ``On the controllability of a class of
  discrete bilinear systems,'' \emph{Automatica}, vol.~9, no.~5, pp. 615--622,
  1973.

\bibitem{LT-KC:11}
L.~Tie and K.~Cai, ``On near-controllability and stabilizability of a class of
  discrete-time bilinear systems,'' \emph{Systems \& Control Letters}, vol.~60,
  no.~8, pp. 650--657, 2011.

\bibitem{WG-JM:98}
W.~Gray and J.~Mesko, ``Energy functions and algebraic {G}ramians for bilinear
  systems,'' in \emph{Preprints of the 4th IFAC Nonlinear Control Systems
  Design Symposium}, 1998, pp. 103--108.

\bibitem{EV:08}
E.~Verriest, ``Time variant balancing and nonlinear balanced realizations,'' in
  \emph{Model Order Reduction: Theory, Research Aspects and Applications},
  2008, pp. 213--250.

\bibitem{PB-TD:11}
P.~Benner and T.~Damm, ``Lyapunov equations, energy functionals, and model
  order reduction of bilinear and stochastic systems,'' \emph{SIAM Journal on
  Control and Optimization}, vol.~49, no.~2, pp. 686--711, 2011.

\bibitem{LZ-JL-BH-GY:03}
L.~Zhang, J.~Lam, B.~Huang, and G.~Yang, ``On {G}ramians and balanced
  truncation of discrete-time bilinear systems,'' \emph{International Journal
  of Control}, vol.~76, no.~4, pp. 414--427, 2003.

\bibitem{SB-MB-US:94}
S.~AL-Baiyat, M.~Bettayeb, and U.~AL-Saggaf, ``New model reduction scheme for
  bilinear systems,'' \emph{International Journal of Systems Science}, vol.~25,
  no.~10, pp. 1631--1642, 1994.

\bibitem{WG-EV:06}
W.~Gray and E.~Verriest, ``Algebraically defined {G}ramians for nonlinear
  systems,'' in \emph{45th IEEE Conference on Decision and Control}, 2006, pp.
  3730--3735.

\bibitem{LZ-JL:02}
L.~Zhang and J.~Lam, ``On $\protect{H}_{2}$ model reduction of bilinear
  systems,'' \emph{Automatica}, vol.~38, no.~2, pp. 205--216, 2002.

\bibitem{PB-TB-TD:11}
P.~Benner, T.~Breiten, and T.~Damm, ``Generalised tangential interpolation for
  model reduction of discrete-time {MIMO} bilinear systems,''
  \emph{International Journal of Control}, vol.~84, no.~8, pp. 1398--1407,
  2011.

\bibitem{MP-RW-JL:13}
M.~Petreczky, R.~Wisniewski, and J.~Leth, ``Balanced truncation for linear
  switched systems,'' \emph{Nonlinear Analysis: Hybrid Systems}, vol.~10, pp.
  4--20, 2013.

\bibitem{TK:80}
T.~Kailath, \emph{Linear Systems}.\hskip 1em plus 0.5em minus 0.4em\relax
  {Englewood Cliffs, New Jersey}: Prentice-Hall, 1980.

\bibitem{RGA-EIJ:71}
R.~G. Agniel and E.~I. Jury, ``Almost sure boundedness of randomly sampled
  systems,'' \emph{SIAM Journal on Control}, vol.~9, no.~3, pp. 372--384, 1971.

\bibitem{AH-RES:87}
A.~Hotz and R.~E. Skelton, ``Covariance control theory,'' \emph{International
  Journal of Control}, vol.~46, no.~1, pp. 13--32, 1987.

\bibitem{RES-SMK-EY:91}
R.~E. Skelton, S.~M. Kherat, and E.~Yaz, ``Covariance control of discrete
  stochastic bilinear systems,'' in \emph{{A}merican {C}ontrol {C}onference},
  Boston, MA, USA, 1991, pp. 2660--2664.

\bibitem{ZA-ZG:94}
Z.~Aganovic and Z.~Gajic, ``The successive approximation procedure for
  finite-time optimal control of bilinear systems,'' \emph{IEEE Transactions on
  Automatic Control}, vol.~39, no.~9, pp. 1932--1935, 1994.

\bibitem{GYT-HM-BLZ:05}
G.~Y. Tang, H.~Ma, and B.~L. Zhang, ``Successive-approximation approach of
  optimal control for bilinear discrete-time systems,'' in \emph{IEE
  Proceedings-Control Theory and Applications}, vol. 152, no.~6, 2005, pp.
  639--644.

\bibitem{EH-BT:88}
E.~Hofer and B.~Tibken, ``An iterative method for the finite-time
  bilinear-quadratic control problem,'' \emph{Journal of Optimization Theory
  and Applications}, vol.~57, no.~3, pp. 411--427, 1988.

\bibitem{SB-LV:09}
S.~Boyd and L.~Vandenberghe, \emph{Convex Optimization}.\hskip 1em plus 0.5em
  minus 0.4em\relax Cambridge University Press, 2009.

\bibitem{TH-SM:84}
T.~Hinamoto and S.~Maekawa, ``Approximation of polynomial state-affine
  discrete-time systems,'' \emph{IEEE Transactions on Circuits and Systems},
  vol.~31, no.~8, pp. 713--721, 1984.

\bibitem{GG-BS:89}
G.~Gallo and B.~Simeone, ``On the supermodular knapsack problem,''
  \emph{Mathematical Programming}, vol.~45, no. 1-3, pp. 295--309, 1989.

\bibitem{RAH-CRJ:85}
R.~A. Horn and C.~R. Johnson, \emph{Matrix Analysis}.\hskip 1em plus 0.5em
  minus 0.4em\relax Cambridge University Press, 1985.

\end{thebibliography}
%
%

\end{document}